\documentclass[12pt]{article}
\usepackage[normalem]{ulem}
\usepackage{amssymb,amsmath,amsthm,latexsym,amsfonts,amscd,dsfont,enumerate}
\usepackage{graphicx}
\usepackage[english]{babel}
\usepackage{booktabs}

\numberwithin{equation}{section}

\newcommand{\re}{\mathrm{e}}
\newcommand{\rd}{\mathrm{d}}
\newcommand{\payoff}{\mathrm{F}}
\newcommand{\indicator}{1\!\!1}
\newcommand{\TPM}[2]{\Pi_{#1,#2}}
\newcommand{\MP}[1]{\mathrm{P}_{#1}}
\newcommand{\RMQA}{RMQA~}
\newcommand{\LTSA}{LTSA~}
\newcommand{\Phiz}{\Phi}
\newcommand{\phiz}{\varphi}

\newenvironment{definition}[1][Definition]{\begin{trivlist}
\item[\hskip \labelsep {\bfseries #1}]}{\end{trivlist}}

\newcommand{\ra}[1]{\renewcommand{\arraystretch}{#1}}

\newtheorem{assumption}{Assumption}
\newtheorem{rem}{Remark}
\newtheorem{proposition}{Proposition}

\newtheorem{ex}{Example}
\newtheorem{lem}{Lemma}

\begin{document}

\title{A backward Monte Carlo approach\\ to exotic option pricing}
\author{Giacomo Bormetti$^{\textrm{a}}$, Giorgia Callegaro$^{b}$, Giulia Livieri$^{\textrm{c,}}$\footnote{Corresponding author. \textit{E-mail address}: giulia.livieri@sns.it}~,\\ and Andrea Pallavicini$^{\textrm{d,e}}$}
\date{
 \today
}

\maketitle

\small
\begin{center}
  $^\textrm{a}$~\emph{Department of Mathematics, 	University of Bologna, Piazza di Porta San Donato 5, 40126 Bologna, Italy}\\
    $^\textrm{b}$~\emph{Department of Mathematics, University of Padova, via Trieste 63, 35121 Padova, Italy}\\
  $^\textrm{c}$~\emph{Scuola Normale Superiore, Piazza dei Cavalieri 7, 56126 Pisa, Italy}\\
  $^\textrm{d}$~\emph{Department of Mathematics, Imperial College, London SW7 2AZ, United Kingdom}\\
  $^\textrm{e}$~\emph{Banca IMI, Largo Mattioli 3, 20121 Milano, Italy}
\end{center}
\normalsize

\smallskip

\begin{abstract}

  \noindent We propose a novel algorithm which allows to sample paths from an underlying price process in a local volatility model and to achieve a substantial variance reduction when pricing exotic options. The new algorithm relies on the construction of a discrete multinomial tree.  The crucial feature of our approach is that -- in a similar spirit to the Brownian Bridge -- each random path runs backward from a terminal fixed point to the initial spot price. We characterize the tree in two alternative ways: in terms of the optimal grids originating from the Recursive Marginal Quantization algorithm and following an approach inspired by the finite difference approximation of the diffusion's infinitesimal generator. We assess the reliability of the new methodology comparing the performance of both approaches and benchmarking them with competitor Monte Carlo methods.  
 \end{abstract}

\newpage
{\small \tableofcontents}
\vfill
\newpage

\section{Introduction}
Pricing financial derivatives typically requires to solve two main issues. In the first place, the choice of a flexible model for the stochastic evolution of the underlying asset price. At this point, a common trade-off arises as models which describe the historical dynamics of the asset price with adequate realism are usually unable to precisely match volatility smiles observed in the option market~\cite{bouchaud2003theory, gatheral2011volatility}. Secondly, once a reasonable candidate has been identified, there is the need to develop fast, accurate, and possibly flexible numerical methods ~\cite{hull2006options,wilmott1993option,clewlow1996implementing}. As regards the former point, Local Volatility (LV) models have become very popular since their introduction by Dupire~\cite{dupire1994pricing}, and Derman and co-authors~\cite{derman1996riding}. Even though the legitimate use of LV models for the description of the asset dynamics is highly questionable, the ability to self-consistently reproduce volatility smiles implied by the market motivates their widespread diffusion among practitioners. Since calibration \emph{\`a la Dupire}~\cite{dupire1994pricing} of LV models assumes the unrealistic availability of a continuum of vanilla option prices across different strikes and maturities~\cite{kahale2004arbitrage}, recent years have seen the emergence of a growing strand of literature dealing with this problem (see for instance \cite{kahale2004arbitrage, coleman1999reconstructing,andreasen2011volatility,lipton2011filling,reghai2012local,pallavicini2015}). In the present paper we fix the calibration following the latest achievements, and we solely focus on the latter issue. Specifically, our goal is to design a novel pricing algorithm based on the Monte Carlo approach able to achieve a sizeable variance reduction with respect to competitor approaches.

\medskip

The main result of this paper is the development of a flexible and efficient pricing algorithm -- termed the \emph{backward Monte Carlo algorithm} -- which runs backward on top of a multinomial tree. The flexibility of this algorithm permits to price generic payoffs without the need of designing tailor-made solutions for each payoff specification. This feature is inherited directly from the Monte Carlo approach (see~\cite{glasserman2003monte} for an almost exhaustive survey of Monte Carlo methods in finance). The efficiency, instead, is linked primarily to the backward movement on the multinomial tree. Indeed, our approach combines both advantages of stratified sampling Monte Carlo and the Brownian Bridge construction~\cite{glasserman2003monte,lapeyre2003understanding,bormetti2006pricing}, extending them to more general financial-asset dynamics than the simplistic assumptions of Black, Scholes, and Merton~\cite{black1973pricing,merton1973theory}. The second purpose of this paper -- minor in relative terms with respect to the first one -- is to investigate an alternative scheme for the implementation of the Recursive Marginal Quantization Algorithm (henceforth \RMQA). The \RMQA is a recursive algorithm which allows to approximate a continuous time diffusion by means of a discrete-time Markov Chain defined on a finite grid of points. The alternative scheme, employed at each step of the \RMQA, is based on the Lloyd I method~\cite{kieffer1982exponential} in combination with the Anderson acceleration Algorithm~\cite{anderson1965iterative,walker2011anderson} developed to solve fixed-point problems. The accelerated scheme permits to speed up the linear rate of convergence of the Lloyd I method~\cite{kieffer1982exponential}, and besides, to fix some flaws of previous \RMQA implementations highlighted in~\cite{callegaro2015pricing}. 

In more detail, a discrete-time Markov Chain approximation of the asset price dynamics can be achieved by introducing at each time step two quantities: (i) a grid for the possible values that the can take, and (ii) the transition probabilities to propagate from one state to another state. Among the approaches discussed in the literature for computing these quantities, in the present paper we analyse and extend two of them. The first approach quantizes via the \RMQA the Euler-Maruyama approximation of the Stochastic Differential Equation (SDE) modelling the underlying asset price. The \RMQA has been introduced in~\cite{sagna2013recursive} to compute vanilla call and put options prices in a pseudo Constant Elasticity of Variance (CEV) LV model. In~\cite{callegaro2015pricing} authors employ it to calibrate a Quadratic Normal LV model. The alternative approach, instead, discretises in an appropriate way the infinitesimal Markov generator of the underlying diffusion by means of a finite difference scheme (see~\cite{kushner2013numerical,albanese2007convergence} for a detailed discussion of theoretical convergence results). We name the latter approach Large Time Step Algorithm, henceforth \LTSA. In~\cite{reghai2012local} authors implement a modified version of the \LTSA to price discrete look-back options in a CEV model, whereas in~\cite{albanese2009spectral} they employ the \LTSA idea to price a particular class of path-dependent payoffs termed Abelian payoffs. More specifically, they incorporate the path-dependency feature -- in the specific case whether or not the underlying asset price hits a specified level over the life of the option -- within the Markov generator. The joint transition probability matrix is then recovered as the solution of a payoff specific matrix equation. The \RMQA and \LTSA present two major differences which can be summarized as follows: (i) the \RMQA permits to recover the optimal -- according to a specific criterion~\cite{printems2005functional} -- multinomial grid, whereas the \LTSA works on a a priori user-specified grid, (ii) the \LTSA necessitates less computational burden than the \RMQA when pricing financial derivatives products whose payoff requires the observation of the underlying on a predefined finite set of dates. Unfortunately, this result holds only for a piecewise time-homogeneous Local Volatility dynamics.

The usage in both equity and foreign exchange (FX) markets of LV models is largely motivated by the flexibility of the approach which allows the exact calibration to the whole volatility surface. Moreover, the accurate re-pricing of plain vanilla instruments and of most liquid European options, together with the stable computation of the option sensitivity to model parameters and the availability of specific calibration procedures, make the LV modelling approach a popular choice. The LV models are also employed in practice to evaluate Asian options and other path-dependent options, although more sophisticated Stochastic Local Volatility (SLV) models are usually adopted. We refer to ~\cite{ren2007calibrating} for details. The price of path-dependent derivative products is then computed either solving numerically a Partial Differential Equation (PDE) or via Monte Carlo methods. The PDE approach is computationally efficient but it requires the definition of a payoff specific pricing equation (see~\cite{wilmott1993option} for an extensive survey on PDE approached in a financial context). Moreover, some options with exotic payoffs and exercise rules are tricky to price even within the Black, Scholes, and Merton framework. On the other hand, standard Monte Carlo method suffers from some inefficiency -- especially when pricing out-of-the-money (OTM) options -- since a relevant number of sampled paths does not contribute to the option payoff. However, the Monte Carlo approach is extremely flexible and several numerical techniques have been introduced to reduce the variance of the Monte Carlo estimator~\cite{clewlow1996implementing,glasserman2003monte}. The backward Monte Carlo algorithm pursues this task.

In this paper we consider the FX market, where we can trade spot and forward contracts along with vanilla and exotic options. In particular, we model the EUR/USD rate using a LV dynamics. The calibration procedure is the one employed in~\cite{reghai2012local} for the equity market and in~\cite{pallavicini2015} for the FX market. Specifically, we calibrate the stochastic dynamics for the EUR/USD rate in order to reproduce the observed implied volatilities with a one basis point tolerance while the extrapolation to implied volatilities for maturities not quoted by the market is achieved by means of a piecewise time-homogeneous LV model. In order to show the competitive performances of the backward Monte Carlo algorithm we compute the price of different kinds of options. We do not price basket options, but we only focus on derivatives written on a single underlying asset considering Asian calls, up-out barrier calls, and auto-callable options. We show that these instruments can be priced more effectively by simulating the discrete-time Markov Chain approximation of the diffusive dynamics from the maturity back to the initial date. In these cases, indeed, the backward Monte Carlo algorithm leads to a significant reduction of the Monte Carlo variance. We leave to a future work the extension of our analysis to SLV models.

The rest of the paper is organized as follows. In Section~\ref{sec:Monte_Carlo} we introduce the key ideas of the backward Monte Carlo algorithm on a multinomial tree. Section~\ref{sec:transition} presents the alternative schemes of implementation based on the \RMQA and \LTSA, and details the numerical investigations testing the performance of both approaches. Section \ref{sec:Numerical} presents a piecewise time-homogeneous LV model for the FX market and reports the pricing performances of the backward Monte Carlo algorithm, benchmarking them with different Monte Carlo algorithms. Finally, in Section \ref{sec:conclusion} we conclude and draw possible perspectives.

\medskip

\noindent \textbf{Notation}: we use the symbol ``$\doteq$'' for the definition of new quantities.

\section{The backward Monte Carlo algorithm}\label{sec:Monte_Carlo}

First of all, let us introduce our working framework. We consider a probability space $(\Omega, \mathcal{F}, \mathbb{P})$, a given time horizon $T >0$ and a stochastic process $X={(X_{t})}_{t \in [0,T]}$ describing the evolution of the asset price. We suppose that the market is complete, so that under the unique risk-neutral probability measure $\mathbb{Q}$, $X$ has a Markovian dynamics described by the following SDE
\begin{equation}\label{eq:dynamic}
  \begin{cases}
    \rd X_{t}=b(t,X_{t})\,\rd t+\sigma(t,X_{t})\,\rd W_{t}\,,\\
    X_{0}=x_{0}\in \mathbb{R}_{+}\,,
  \end{cases}
\end{equation}
where $(W_{t})_{t\in [0,T]}$ is a standard one-dimensional $\mathbb{Q}$-Brownian motion, and $b:[0,T]\times \mathbb{R}_{+}\rightarrow \mathbb{R}$ and $\sigma:[0,T]\times \mathbb{R}_{+}\rightarrow \mathbb{R}_{+}$ are two measurable functions satisfying the usual conditions ensuring the existence and uniqueness of a (strong) solution to the SDE \eqref{eq:dynamic}. Besides, we consider deterministic interest rates. 
\noindent Henceforth, we will always work under the risk-neutral probability measure $\mathbb{Q}$, since we focus on the pricing of derivative securities written on $X$. Specifically, we are interested in pricing financial derivative products whose payoff may depend on the whole path followed by the underlying asset, i.e. path-dependent options.

\medskip

Let us now motivate the introduction of our novel pricing algorithm. Even in the classical Black, Scholes, and Merton~\cite{black1973pricing,merton1973theory} framework, when pricing financial derivatives the actual analytical tractability is limited to plain vanilla call and put options and to few other cases (for instance, see the discussion in \cite{hui2000comment,vecer2004pricing}). Such circumstances motivate the quest for general and reliable pricing algorithms able to handle more complex contingent claims in more realistic stochastic market models. In this respect the Monte Carlo (MC) approach represents a natural candidate. Nevertheless, a general purpose implementation of the MC method is known to suffer from a low rate of convergence. In particular, in order to increase its numerical accuracy, it is either necessary to draw a large number of paths
or to implement tailor-made variance reduction techniques. Moreover, the standard MC estimator is strongly inefficient when considering out-of-the-money (OTM) options, since a relevant fraction of sampled paths does not contribute to the payoff function. 
\noindent For these reasons, we present a novel MC methodology which allows to effectively reduce the variance of the estimated price. To this end, we proceed as follows. First we introduce a discrete-time and discrete-space process $\widehat{\bar{X}}$ approximating the continuous time (and space) process $X$ in Equation \eqref{eq:dynamic}. Then, we propose a MC approach -- the backward Monte Carlo algorithm -- to sample paths from $\widehat{\bar{X}}$  and to compute derivative prices.

\medskip
In particular, we first split the time interval $[0,T]$ into $n$ equally-spaced subintervals $[t_{k}, t_{k+1}]$, $k \in \{ 0,\dots,n-1 \}$, with $t_{0}=0$, $t_{n}=T$ and we approximate the SDE in Equation \eqref{eq:dynamic} with an Euler-Maruyama scheme as follows:
\begin{equation}\label{eq:Euler_Maruyama}
  \begin{cases}
    \bar{X}_{t_{k+1}}=\bar{X}_{t_{k}}+b(t_{k},\bar{X}_{t_{k}})\Delta t+\sigma(t_{k},\bar{X}_{t_{k}})\sqrt{\Delta t} \ Z_k\,,\\
    \bar{X}_{t_{0}}=X_0= x_{0}\,,
  \end{cases}
\end{equation}
where $(Z_{k})_{0 \le k \le n-1}$ is a sequence of $i.i.d.$ standard Normal random variables and $\Delta t \doteq t_{k+1}-t_{k} = T/n$.  
\noindent Then, we assume that $\forall k \in \{1, \dots, n\}$ each random variable $\bar{X}_{t_{k}}$ in Equation \eqref{eq:Euler_Maruyama} can be approximated by a discrete random variable taking values in $\Gamma_{k}\doteq \{\gamma_{1}^{k},\dots, \gamma_{N}^{k}\}$, whereas for $t_0$ we have $\Gamma_{0}=\gamma^{0}=x_0$. We denote by $\widehat{\bar{X}}_{t_{k}}$ the discrete-valued approximation of the random variable $\bar{X}_{t_{k}}$. In this way we constrain the discrete-time Markov process $(\widehat{\bar{X}}_{t_{k}})_{1 \le k \le n}$ to live on a multinomial tree. Notice that, by definition $\vert \Gamma_{k} \vert = N$, all $k \in \{1, \dots, n\}$. Nevertheless, this is not the most general setting. For instance, within the \RMQA framework authors in ~\cite{sagna2013recursive} perform numerical experiments letting the number of points in the space discretisation grids vary over time. However, they underline how the complexity in the time varying case becomes higher as $N$ increases, although the difference in the results is negligible.

\noindent In order to define our pricing algorithm, we need the transition probabilities from a node at time $t_k$ to a node at time $t_{k+1}$, $k\in\{0,\dots,n-1\} $, so that in the next section we provide a detailed description of two different approaches to consistently approximate them. For the moment, we describe the backward Monte Carlo algorithm assuming the knowledge of both the multinomial tree $(\Gamma_{k})_{0 \le k\le n}$ and the transition probabilities.

\medskip

As aforementioned, our final target  is the computation at time $t_{0}$ of the fair price of a path-dependent option with maturity $T>0$. We denote by $\payoff$ its general discounted payoff function. In particular, it is a function of a finite number of discrete observations. We are not going to make precise $\payoff$ at this point, we only recall here that in this paper we will focus on Asian options, up-and-out barrier options and auto-callable options.
According to the arbitrage pricing theory~\cite{bjork2004arbitrage}, the price is given by the conditional expectation of the discounted payoff under the risk-neutral measure $\mathbb Q$, given the information available at time $t_0$. By means of the Euler-Maruyama discretisation, we can approximate the option price as follows:
\begin{equation*}
\mathbb{E}_{t_{0}}\left[\payoff(x_0, \bar X_{t_1}, \dots,\bar{X}_{t_{n}})\right]=\int_{\mathbb{R}^{n}} \payoff(x_0,x_1,\dots,x_n) \ p(x_0,x_1,\dots,x_n) \ \rd x_1\cdots \rd x_n \,,
\end{equation*}
where $p(x_0,x_1,\dots,x_n)$ is the joint probability density function (PDF) of $(\bar X_0, \bar X_{t_1}, \dots, \bar X_{t_n})$. The previous expression can be further approximated exploiting the process $\widehat{\bar{X}}$ and its discrete nature (recall that $\Gamma_{k} = \{\gamma_{1}^{k}, \dots, \gamma_{N}^{k} \}$):

\begin{equation}\label{eq:price_1}
\begin{split}
\mathbb{E}_{t_{0}}\left[\payoff(x_0,\bar X_{t_1},\dots, \bar{X}_{t_{n}})\right] & \simeq \mathbb{E}_{t_{0}}\left[ \payoff(x_0,\widehat{\bar{X}}_{t_{1}},\dots, \widehat{\bar{X}}_{t_{n}} ) \right]\\
&= \sum_{i_{1}=1}^{N}\dots \sum_{i_{n}=1}^{N}~\payoff(x_0,\gamma_{i_{1}}^{1}\dots,\gamma_{i_n}^{n}) \ \mathbb{P}(x_0,\gamma_{i_{1}}^{1},\dots,\gamma_{i_{n}}^{n}),
\end{split}
\end{equation}
where 
\begin{equation*}
\mathbb{P}(x_0,\gamma_{i_{1}}^{1},\dots,\gamma_{i_{n}}^{n}) \doteq \mathbb P (\widehat{\bar{X}}_{t_{0}} = x_0, \widehat{\bar{X}}_{t_{1}} = \gamma^1_{i_1}, \dots, \widehat{\bar{X}}_{t_{n}} = \gamma^n_{i_n}).
\end{equation*}
Exploiting the Markovian nature of $\widehat{\bar{X}}$ and using Bayes' theorem, we rewrite the right hand side of Equation \eqref{eq:price_1} in the following, equivalent, way:
\begin{equation}\label{eq:price_2}
\sum_{i_{1}=1}^{N}\dots \sum_{i_{n}=1}^{N}~\payoff(x_0,\gamma^1_{i_{1}},\dots,\gamma^n_{i_n}) \ \mathbb{P}(\gamma^1_{i_{1}}\vert x_0) \cdots \mathbb{P}(\gamma^n_{i_{n}}\vert \gamma^{n-1}_{i_{n-1}}),
\end{equation}
where
\begin{equation}
\mathbb{P}(\gamma_{i_{k+1}}^{k+1} \vert \gamma_{i_{k}}^{k}) \doteq \mathbb P(\widehat{\bar{X}}_{t_{k+1}} = \gamma_{i_{k+1}}^{k+1} \vert \widehat{\bar{X}}_{t_{k}} = \gamma_{i_{k}}^{k}),
\end{equation}
all $\gamma_{i_{k}}^{k}\in \Gamma_{k}$ and all $k \in \{1, \dots, n-1 \}$

\medskip

In order to compute the expression in Equation \eqref{eq:price_2}, a straightforward application of the standard MC theory would require the simulation of $N_{MC}$ paths all originating from $x_0$ at time $t_{0} = 0$. The same aforementioned arguments about the lack of efficiency of the MC estimator for the case of a continuum of state-spaces still hold for the discrete case. However, forcing each random variable $\bar{X}_{t_{k}}$, $1 \le k \le n$, to take at most $N$ values leads in general to a reduction of the variance of the Monte Carlo estimator.

\medskip

For each $t_k\in \{t_1,\dots,t_{n-1}\}$ we denote by $\Pi^{k, k+1}$ the $(N \times N)$-dimensional matrix whose elements are the transition probabilities:
\begin{equation*}
\Pi_{i,j}^{k, k+1} \doteq \mathbb{P}(\gamma^{k+1}_{j}\vert \gamma^k_{i}),\quad \gamma_i^k \in \Gamma_k,\text{ } \gamma_j^{k+1}\in \Gamma_{k+1},\text{ and  }i, j \in \{1, \dots, N \}.
\end{equation*}
The key idea behind the backward Monte Carlo algorithm is to express $\Pi_{i,j}^{k+1,k}$ as a function of $\Pi_{i,j}^{k, k+1} $ by applying Bayes' theorem:
\begin{equation}\label{eq:bayes}
\TPM{i}{j}^{k+1,k}=\frac{\TPM{i}{j}^{k, k+1}~\MP{i}^k}{\MP{j}^{k+1}} 
\end{equation}
where ${\MP{i}^{k}} \doteq \mathbb P(\widehat{\bar{X}}_{t_{k}} = \gamma_i^k\vert \bar{X}_{t_{0}} = x_{0})$ and ${\MP{j}^{k+1}} \doteq \mathbb P(\widehat{\bar{X}}_{t_{k+1}} = \gamma_j^{k+1} \vert \bar{X}_{t_{0}} = x_{0})$. Iteratively, we recover all the transition probabilities  which allow us to go through the multinomial tree in a backward way from each terminal point to the initial node $x_0$. In particular, relation in Equation \eqref{eq:bayes} permits to re-write the joint probability appearing in Equation\eqref{eq:price_1} and then in Equation \eqref{eq:price_2}  as 
\begin{equation*}
  \mathbb{P}(x_0,\gamma_{i_{1}}^{1},\dots,\gamma_{i_{n}}^{n})= \mathbb{P}(\gamma^1_{i_{1}}\vert x_0) \cdots \mathbb{P}(\gamma^n_{i_{n}}\vert \gamma^{n-1}_{i_{n-1}}) = \left(\prod_{k=0}^{n-1} \TPM{i_{k+1}}{i_{k}}^{k+1,k} \right) \MP{i_{n}}^n.
\end{equation*}
Consistently, we obtain the following proposition, containing the core of our pricing algorithm:\begin{proposition}
The price of a path-dependent option with discounted payoff $\payoff	$, \\$\mathbb{E}_{t_{0}}\left[\payoff(x_0,\bar X_{t_1},\dots, \bar{X}_{t_{n}})\right]$, can be approximated by:
\begin{equation*}
  \begin{split}
   \mathbb{E}_{t_{0}}\left[ \payoff(x_0,\widehat{\bar{X}}_{t_{1}},\dots, \widehat{\bar{X}}_{t_{n}} ) \right] 
& = \sum_{i_{n}=1}^{N} \MP{i_{n}}^n \sum_{i_{1}=1}^{N}\dots\sum_{i_{n-1}=1}^{N}\left(\prod_{k=0}^{n-1}\TPM{i_{k+1}}{i_{k}}^{k+1,k}\right)\payoff(x_0,\gamma^1_{i_{1}},\dots,\gamma^n_{i_n})\\
    &\doteq \sum_{i_{n}=1}^{N}\MP{i_{n}}^n \mathcal{F}(x_0,\gamma_{i_{n}}^{n}),
  \end{split}
\end{equation*}

\noindent where $\mathcal{F}(x_0,\gamma_{i_{n}}^{n})$ is the expectation of the payoff function $\payoff$ with respect to all paths starting at $x_0$ and terminating at $\gamma^n_{i_n}$. 
\end{proposition}
The expectation $\mathcal{F}(x_0,\gamma_{i_{n}}^{n})$ can be computed sampling $N_{MC}^{i_{n}}$ MC paths from the conditional law of $(\widehat{\bar{X}}_{t_{1}},\dots,\widehat{\bar{X}}_{t_{n-1}})$ given $x_0$ and $\gamma^n_{i_n}$, thus obtaining, at the same time, the error $\sigma_{i_{n}}$ associated to the MC estimator. By virtue of the Central Limit Theorem the errors scale with the square root of $N_{MC}^{i_{n}}$, so that the larger $N_{MC}^{i_{n}}$ is, the smaller the error. In particular, if we indicate with $\widetilde{\Gamma}_{n}=\{\tilde{\gamma}_{1}^{n},\dots,\tilde{\gamma}_{N^+}^{n}\}$, with $N^+\leq N$, those points of $\Gamma_{n}$ for which the payoff $\payoff$ is different from zero, we estimate the boundary values corresponding to the $95\%$ confidence interval for the derivative price as
\begin{equation}\label{eq:price_4}
  \sum_{i_{n}=1}^{N^+}\MP{i_{n}}^n \widehat{\mathcal{F}(x_0,\tilde{\gamma}_{i_{n}}^{n})}\pm 1.96 \sqrt{\sum_{i_{n}=1}^{N^+}(\MP{i_{n}}^n \sigma_{i_{n}})^{2}}\,,
\end{equation}
where $\widehat{\mathcal{F}(x_0,\tilde{\gamma}_{i_{n}}^{n})}$ corresponds to the Monte Carlo estimator of $\mathcal{F}(x_0,\tilde{\gamma}_{i_{n}}^n)$. It is worth noticing that the error in Equation~\eqref{eq:price_4} does not take into account the effect of finiteness of $\widetilde{\Gamma}_{n}$.\\
\noindent A sizeable variance reduction results from having split the $n$ sums in Equation \eqref{eq:price_2} into the external summation over the points of the deterministic grid $\widetilde{\Gamma}_{n}$ and the evaluation of an expectation of the payoff with fixed initial and terminal points. This procedure corresponds to the variance reduction technique known as stratified sampling MC \cite{glasserman2003monte}. In particular, in \cite{glasserman2003monte} authors prove analytically that the variance of the MC estimator without stratification is always greater than or equal to that of the stratified one. As pointed out in \cite{glasserman2003monte} stratified sampling involves consideration of two issues: (i) the choice of the points in $\Gamma_{n}$ and the allocation $N_{MC}^{i_n}$, $i_{n} \in \{1, \dots, N\}$, (ii) the generation of samples from $\widehat{\bar{X}}$ conditional on $\widehat{\bar{X}}_{t_{n}} \in \Gamma_{n}$ and on $\widehat{\bar{X}}_{t_{0}} = \gamma_{0}$. Our procedure resolves both these points. Precisely, once selected $\widetilde{\Gamma}_{n}$, the Backward Monte Carlo algorithm allows us to choose the number of paths from all the points in $\widetilde{\Gamma}_{n}$, independently on the value of $P_{i_{n}}^{n}$.

\medskip

At this point, two are the main ingredients needed in order to compute the quantities in Equation \eqref{eq:price_4}: (i) the transition probabilities, (ii) the fast backward simulation of the process $\widehat{\bar{X}}$. For both purposes, we introduce ad-hoc numerical procedures. As regards the former point, we analyse and extend two approaches already present in the literature: the first one is based on the concept of optimal state-partitioning of a random variable (called stratification in \cite{barraquand1995numerical} and quantization in \cite{bally2003quantization}) and employs the \RMQA~\cite{sagna2013recursive, callegaro2015pricing}. The second approach provides a recipe to compute in an effective way the transition probability matrix between any two arbitrary dates for a piecewise time-homogeneous process~\cite{reghai2012local,albanese2007operator}. More details on these two methods will be given in Section \ref{sec:transition}.

\medskip

For what concerns the backward simulation, we employ the Alias method introduced in~\cite{kronmal1979alias}. More specifically, for every $k$ from $n-1$ to $1$, the (backward) simulation of $\widehat{\bar{X}}_{t_{k}}$ conditional on $\{ \widehat{\bar{X}}_{t_{k+1}}=\gamma^{k+1}_j \}$ is equivalent to sampling at each time $t_{k+1}$ from a discrete non-uniform distribution with support $\Gamma_{k}$ and probability mass function equal to the $j$-th row of $\Pi^{k+1,k}$.
Given the discrete distribution, a na\"{i}ve simulation scheme consists in drawing a uniform random number from the interval $[0,1]$ and recursively search over the cumulative sum of the discrete probabilities. However, in this case the corresponding computational time grows linearly with the number $N$ of states. The Alias method, instead, reduces this numerical complexity to $O(1)$ by cleverly pre-computing a table -- the Alias table -- of size $N$. We base our implementation on this method, which enables a large reduction of the MC computation time. A more detailed description of the Alias method can be found at \url{www.keithschwarz.com}.

\section{Recoverying the transition probabilities}\label{sec:transition}
We present the two approaches used for the approximation of the transition probabilities of a discrete-time Markov Chain. The \RMQA is described and extended in Section~\ref{subsec:rmq_algorithm}. In particular, we first provide a brief overview on optimal quantization of a random variable and then we propose an alternative implementation of the \RMQA. The \LTSA is presented in Section~\ref{subsec:fast_exponentiation}, where we also provide a brief introduction on Markov processes and generators.

\subsection{A quantization based algorithm}\label{subsec:rmq_algorithm}
The reader who is familiar with quantization can skip the following subsection.
\subsubsection{Optimal quantization}
We present here the concept of optimal quantization of a random variable by emphasizing its practical features, without providing all the mathematical details behind it. A more extensive discussion can be found e.g. in~\cite{printems2005functional,graf2000foundations,pages2014introduction,pages2004optimal}.

Let $\bar{X}$ be a one-dimensional continuous random variable defined on a probability space $(\Omega, \mathcal{F}, \mathbb{P})$ and $\mathbb{P}_{\bar X}$ the measure induced by it. The quantization of $\bar{X}$ consists in approximating it by a one-dimensional discrete random variable $\widehat{\bar{X}}$. In particular, this approximation is defined by means of a quantization function $q_{N}$ of $\bar X$, that is to say $\widehat{\bar{X}}\doteq q_{N}(\bar X)$, defined in such a way that $\widehat{\bar{X}}$ takes $N\in \mathbb{N}^{+}$ finitely many values in $\mathbb{R}$. The finite set of values for $\widehat{\bar{X}}$, denoted by $\Gamma \equiv \{\gamma_{1},\dots,\gamma_{N}\}$, is the quantizer of $\bar X$, while the image of the function $q_{N}$ is the related quantization. The components of $\Gamma$ can be used as generator points of a Voronoi tessellation $\{C_{i}(\Gamma)\}_{i=1,\dots,N}$. In particular, one sets up the following tessellation with respect to the absolute value in $\mathbb{R}$
\begin{equation*}
  C_{i}(\Gamma)\subset \{\gamma \in \mathbb{R} : \vert \gamma - \gamma_{i}\vert=\min_{1\le j \le N}\vert \gamma-\gamma_{j}\vert \}\,,
\end{equation*}%
and the associated quantization function $q_{N}$ is defined as follows:
\begin{equation*}
  q_{N}(\bar X)=\sum_{i=1}^{N}\gamma_{i} \indicator_{C_{i}(\Gamma)}(\bar X)\,.
\end{equation*}

Notice that in our setting, we are going to quantize the random variables  ${(\bar X_{t_k})}_{0 \le k \le n}$ introduced in Equation \eqref{eq:Euler_Maruyama}.

Such a construction rigorously define a probabilistic setting for the random variable $\widehat{\bar{X}}$ , by exploiting the probability measure induced by the continuous random variable $\bar X$. The approximation of $\bar X$ through $\widehat{\bar{X}}$  induces an error, whose $L^2$ version -- called $L^{2}$-mean quantization error -- is defined as
\begin{equation}\label{eq:error}
  \| \bar X-q_{N}(\bar X)\|_{2}\doteq \sqrt{\mathbb{E}\left[\min_{1 \le i \le N}| \bar X-\gamma_{i}|^{2}\right]}\,.
\end{equation}%
The expected value in Equation~\eqref{eq:error} is computed with respect to the probability measure which characterizes the random variable $\bar X$. The purpose of the optimal quantization theory is finding a quantizer\footnote{In one dimension the uniqueness of the optimal $N$ quantizer is guaranteed if the distribution of $\bar X$ is absolutely continuous with a log-concave density function~\cite{pages2014introduction}.}  indicated by $\Gamma^{*}$, which minimizes the error in Equation~\eqref{eq:error} over all possible quantizers with size at most $N$.\\ 
From the theory (see, for instance ~\cite{graf2000foundations}) we know that the mean quantization error vanishes as the grid size $N$ tends to infinity and its rate of convergence is ruled by Zador theorem. However, computationally, finding explicitly $\Gamma^{*}$ can be a challenging task. This has motivated the introduction of sub-optimal criteria linked to the notion of stationary quantizer~\cite{pages2014introduction}:
\begin{definition}
 A quantizer $\Gamma\equiv \{\gamma_{1},\dots,\gamma_{N}\}$ inducing the quantization $q_{N}$ of the random variable $\bar X$ is said to be $stationary$ if
\begin{equation}\label{eq:stationarity}
    \mathbb{E} \left[ \bar X\vert q_{N}(\bar X)\right]=q_{N}(\bar X)\,.
\end{equation}%
\end{definition}

\begin{rem}
An optimal quantizer is stationary, the vice-versa does not hold true in general (see, for instance ~\cite{pages2014introduction}).
\end{rem}

In order to compute optimal (or sub-optimal) quantizers, one first introduces a notion of distance between a random variable $\bar X$ and a quantizer $\Gamma$ \begin{equation*}
  d(\bar X,\Gamma)\doteq  \min_{1 \le i \le N}| \bar X-\gamma_{i}|\,,
\end{equation*} 
and then one considers the so called distortion function
\begin{equation}\label{eq:distortion}
  D(\Gamma)\doteq \mathbb{E}\left[d(\bar X,\Gamma)^{2} \right]=\mathbb{E}\left[\min_{1 \le i \le N}| \bar X-\gamma_{i}|^{2}\right]=\sum_{i=1}^{N}\int_{C_{i}(\Gamma)}\vert \xi-\gamma_{i}\vert^2 \,\rd\mathbb{P}_{\bar X}(\xi)\,.
\end{equation}
It can be shown (see, for instance \cite{pages2014introduction}) that the distortion function is continuously differentiable as a function of $\Gamma$. In particular, it turns out that stationary quantizers are critical points of the distortion function, that is, a stationary quantizer $\Gamma$ is such that $\triangledown D(\Gamma)=0$.\\ 
Several numerical approaches have been proposed in order to find stationary quantizers (for a review see~\cite{pages2004optimal}). These approaches can be essentially divided into two categories: gradient-based methods and fixed-point methods. The former class includes the Newton-Raphson algorithm, whereas the second category includes the Lloyd I algorithm~\cite{kieffer1982exponential}. More specifically, the Newton-Raphson algorithm requires the computation of the gradient, $\triangledown D(\Gamma)$, and of the Hessian matrix, $\triangledown^{2} D(\Gamma)$, of the distortion function. The Lloyd I algorithm, on the other hand, does not require the computation of the gradient and Hessian and it consists in a fixed-point algorithm based on the stationary Equation~\eqref{eq:stationarity}.

\subsubsection{The Recursive Marginal Quantization Algorithm}
The \RMQA is a recursive algorithm, which has been recently introduced by G. Pag\`es and A. Sagna in~\cite{sagna2013recursive}. It consists in quantizing the stochastic process $X$ in Equation \eqref{eq:dynamic} by working on the (marginal) random variables $\bar{X}_{t_{k}}$, all $k \in \left\{1, \dots, n \right\}$ in \eqref{eq:Euler_Maruyama}. The key idea behind the \RMQA is that the discrete-time Markov process $\bar{X} = (\bar{X}_{t_{k}})_{0\le k \le n}$ in Equation \eqref{eq:Euler_Maruyama} is completely characterized by the initial distribution of $\bar{X}_{t_{0}}$ and by the transition probability densities. We indicate by $\widehat{\bar{X}}_{t_{k}}$ the quantization of the random variable $\bar{X}_{t_{k}}$ and by $\bar{D}(\Gamma_{k})$ the associated distortion function.

\medskip

\begin{rem}
  The process $\widehat{\bar{X}} = (\widehat{\bar{X}}_{t_{k}})_{0 \le k \le n}$ is not, in general, a discrete-time Markov Chain. Nevertheless, it is known (see, for instance~\cite{pages2004optimal}) that there exists a discrete-time Markov Chain, $\widehat{\bar{X}}^{c}\doteq(\widehat{\bar{X}}^{c}_{t_{k}})_{0 \le k\le n}$, with initial distribution and transition probabilities equal to those of $\widehat{\bar{X}}$. Hence, throughout the rest of the paper, when we will write ``discrete-time Markov Chain'' within the Recursive Marginal Quantization framework we will refer, by tacit agreement, to the process $\widehat{\bar{X}}^{c}$.
\end{rem}
Here we give a quick drawing of the \RMQA. First of all, one introduces the Euler operator associated to the Euler scheme in Equation \eqref{eq:Euler_Maruyama}:
\begin{equation*}
\mathcal{E}_{k}(x,\Delta t;Z)\doteq x+b(t_{k},x)\Delta+\sigma(t_{k},x)\sqrt{\Delta t} \ Z
\end{equation*}
where $Z	\sim\mathcal{N}(0,1)$, so that, from \eqref{eq:Euler_Maruyama}, $\bar{X}_{t_{k+1}}=\mathcal{E}_{k}(\bar{X}_{t_{k}},\Delta t; Z_k)$. 
\begin{lem}\label{lem:condDistrib}
Conditionally on the event $\{\bar{X}_{t_{k}}= x \}$, the random variable $\bar{X}_{t_{k+1}}$ is a Gaussian random variable with mean $m_{k}(x) = x + b(t_{k},x) \Delta$ and standard deviation $v_{k}(x)=\sqrt{\Delta}\sigma(t_{k},x)$, all $k = 1, \dots, n-1$.
\end{lem}
\begin{proof} It follows immediately from the equality $\bar{X}_{t_{k+1}}=\mathcal{E}_{k}(\bar{X}_{t_{k}},\Delta t; Z_k)$, given that $Z_k$ is a standard normal random variable.
\end{proof}
At this point, one writes down the following crucial equalities:
\begin{equation}\label{eq:distortion_2}
  \begin{split}
    \bar{D}(\Gamma_{k+1})&=\mathbb{E}\left[d(\bar{X}_{t_{k+1}},\Gamma_{k+1})^2\right]=\mathbb{E}\left[\mathbb{E}\left[d(\bar{X}_{t_{k+1}},\Gamma_{k+1})^2\vert\bar{X}_{t_{k}}\right]\right]=\mathbb{E}\left[d(\mathcal{E}_{k}(\bar{X}_{t_{k}},\Delta t; Z_{k}),\Gamma_{k+1})^2\right]\, ,
  \end{split}
\end{equation}
where ${(Z_{k})}_{0\le k \le n}$ is a sequence of \textit{i.i.d.} one-dimensional standard normal random variables. As said, stationary quantizers are zeros of the gradient of the distortion function. By definition, the distortion function $\bar{D}(\Gamma_{k+1})$ depends on the distribution of $\bar{X}_{t_{k+1}}$, which is, in general, unknown. Nevertheless, thanks to Lemma \ref{lem:condDistrib}, the distortion in Equation \eqref{eq:distortion_2} can be computed explicitly.
\noindent Equation ~\eqref{eq:distortion_2} is the essence of the \RMQA. More precisely, one starts setting the quantization of $\bar{X}_{t_{0}}$ to $x_0$, namely $q_{N}(\bar{X}_{t_{0}})=x_{0}$. Then, one approximates $\bar{X}_{t_{1}}$ with $\widetilde{X}_{t_{1}}\doteq \mathcal{E}_{0}(x_{0},\Delta t; Z_{{1}})$ and the distortion function associated to $\bar{X}_{t_{1}}$ with that associated to $\widetilde{X}_{t_{1}}$, namely $\bar{D}(\Gamma_{1})\approx \widetilde{D}(\Gamma_{1})\doteq \mathbb{E}\left[d(\mathcal{E}_{0}(x_{0},\Delta t; Z_{{1}}),\Gamma_{1})^2\right]$. Then, one looks for a stationary quantizer $\Gamma_1$ by searching for a zero of the gradient of the distortion function, using either Newton-Raphson or Lloyd I method. The procedure is applied iteratively at each time step $t_{k}$, $1 \le k \le n$, leading to the following sequence of stationary (marginal) quantizers:
\begin{equation*}
  \begin{split}
    &\widehat{\widetilde{X}}_{t_{0}}\doteq \bar{X}_{t_{0}}\,,\\
    &\widehat{\widetilde{X}}_{t_{k}}=q_{N}(\widetilde{X}_{t_{k}})\quad\mathrm{and}\quad\widetilde{X}_{t_{k+1}}=\mathcal{E}_{k}(\widehat{\widetilde{X}}_{t_{k}},\Delta t; Z_{{k+1}})\,,\\
    &(Z_{{k}})_{1\le k\le n}\quad \text{\textit{i.i.d.} Normal random variables independent from } \bar{X}_{t_{0}}\,.
  \end{split}
\end{equation*}%
In~\cite{sagna2013recursive} the authors give an estimation of the (quadratic) error bound $\|\bar{X}_{t_{k}}-\widehat{\bar{X}}_{t_{k}}\|_{2}$, for fixed $k=1,\dots,n$. 

\noindent
At this point, the approximated transition probabilities (termed companion parameters in \cite{sagna2013recursive}) are obtained instantaneously given the quantization grids and Lemma \ref{lem:condDistrib}. In particular:
\begin{equation}
\Pi^{k,k+1}_{i,j} = \mathbb{P}(\gamma_{j}^{k+1}\vert \gamma_{i}^{k}) \approx \mathbb{P}\left(\widetilde{X}_{t_{k+1}}\in C_{j}(\Gamma_{k+1}) \vert \widetilde{X}_{t_{k}} \in C_{i}(\Gamma_{k})\right).
\end{equation}

In Appendix \ref{sec:appendix_2} we provide the explicit expressions of the distortion function $\widetilde{D}(\Gamma_{k+1})$ and of the approximated transition probabilities $\mathbb{P}\left(\widetilde{X}_{t_{k+1}}\in C_{j}(\Gamma_{k+1}) \vert \widetilde{X}_{t_{k}} \in C_{i}(\Gamma_{k})\right)$.

\medskip

In order to compute numerically the sequence of stationary quantizers $(\Gamma_{k})_{1\le k \le n}$ in ~\cite{callegaro2015pricing, sagna2013recursive} authors employ the Newton-Raphson algorithm. However, as pointed out in~\cite{callegaro2015pricing}, it may become unstable when $\Delta t \rightarrow 0$ due to the ill-condition number of the Hessian matrix $\triangledown^{2} D(\Gamma)$. An alternative approach is based on fixed-point algorithms, such as the Lloyd I method, even though such method converges to the optimal solution with a smaller rate of convergence (see~\cite{kieffer1982exponential} for a discussion). For these reason and as original contribution we combine it with a particular acceleration scheme, called Anderson Acceleration.

\subsubsection{The Anderson accelerated procedure}
The acceleration scheme, called \textit{Anderson} \textit{acceleration}, was originally discussed in Anderson~\cite{anderson1965iterative}, and outlined in~\cite{walker2011anderson} together with some practical considerations for implementations. For completeness, in Appendix \ref{sec:appendix_1} we give some details on how the Lloyd I method works when employed in the RMQ setting.

\medskip

Now, we discuss the major differences between a general fixed-point algorithm -- and its associated fixed-point iterations -- and the same fixed-point method coupled with the Anderson acceleration. We outline the practical features without giving all the technical details concerning the numerical implementation of the accelerated scheme (please refer to~\cite{walker2011anderson} for an extensive discussion on this issue). 

A general fixed-point problem -- also known as Picard problem -- and its associated fixed-point iteration are defined as follows:
\begin{equation}\label{eq:fixed_point}
  \begin{split}
    \textbf{Fixed-point problem}:&\,\,\text{Given}\,\,g:\mathbb{R}^{N}\rightarrow \mathbb{R}^{N}\,,\quad\text{find} \ \Gamma \in \mathbb R^N \ \text{s.t.} \quad \Gamma=g(\Gamma).\\
    \emph{\textbf{Algorithm}}&\,\,\textbf{(Fixed}\,\,\textbf{Point}\,\,\textbf{Iteration)}\\
    &\text{Given}\quad\Gamma^{0},\\
    &\text{for}\,\,l \ge 0, l \in \mathbb N \\
    &\,\,\,\,\,\,\,\,\,\,\,\,\text{set}\,\,\Gamma^{l+1}=g(\Gamma^{l})\,.
  \end{split}
\end{equation}
The same problem coupled with the Anderson acceleration scheme is modified as follows:
\begin{equation}\label{eq:fixed_point_acc}
  \begin{split}
    &\emph{\textbf{Algorithm}}\,\,\textbf{(Anderson}\,\,\textbf{acceleration)}\\
    &\text{Given}\,\,\Gamma^{0}\,\,\text{and}\,\,m\ge 1\,\,, m\in \mathbb{N} ,\\
    &\text{set}\,\,\Gamma^{1}=g(\Gamma^{0}),\\
    &\text{for}\,\,l \ge 1, l \in \mathbb N\\
    &\,\,\,\,\,\,\,\,\,\,\,\,\text{set}\,\,m_{l} =\min(m,l)\\
    &\,\,\,\,\,\,\,\,\,\,\,\,\text{set}\,\,F_{l}=(f_{l-m_{l}},\dots,f_{l}),\,\,\text{where}\,\,f_{i}=g(\Gamma^{i})-\Gamma^{i}\\
    &\,\,\,\,\,\,\,\,\,\,\,\,\text{determine}\,\,\alpha^{(l)}=(\alpha_{0}^{(l)},\dots,\alpha_{m_l}^{(l)})^{T}\,\,\text{that solves}\\
    &\,\,\,\,\,\,\,\,\,\,\,\,\,\,\,\,\,\,\,\,\,\,\,\,\min_{\alpha^{l}\in \mathbb{R}^{m_{l}+1}}\|F_{l}\alpha^{(l)}\|_{2}\,\,\text{s.t.}\,\,\sum_{i=0}^{m_{l}}\alpha_{i}^{(l)}=1\\
    &\,\,\,\,\,\,\,\,\,\,\,\,\text{set}\,\,\Gamma^{l+1}=\sum_{i=0}^{m_{l}}\alpha_{i}^{(l)}g(\Gamma_{l-m_{l}+i}).\\
  \end{split}
\end{equation}
The Anderson acceleration algorithm stores (at most) $m$ user-specified previous function evaluations and computes the new iterate as a linear combination of those evaluations with coefficients minimising the Euclidean norm of the weighted residuals. In particular, with respect to the general fixed-point iteration, Anderson acceleration exploits more information in order to find the new iterate. 

In Equation~\eqref{eq:fixed_point_acc} Anderson acceleration algorithm allows to monitor the conditioning of the least squares problem. In particular, we follow the strategy used in~\cite{walker2011anderson} where the constrained least squares problem is first casted in an unconstrained one, and then solved using a QR decomposition. The usage of the QR decomposition to solve the unconstrained least square problem represents a good balance of accuracy and efficiency. Indeed, if we name $\mathcal{F}_{l}$ the least squares problem matrix, it is obtained from its predecessor $\mathcal{F}_{l-1}$ by adding a new column on the right. The QR decomposition of $\mathcal{F}_{l}$ can be efficiently attained from that of $\mathcal{F}_{l-1}$ in $O(m_{l}N)$ arithmetic operations using standard QR factor-updating techniques (see~\cite{golub2012matrix}).

The Anderson acceleration scheme speeds up the linear rate of convergence of the general fixed-point problem without increasing its computational complexity. More importantly, it does not suffer the extreme sensitivity of the Newton-Raphson method to the choice of the initial point (grid). 
We refer to the numerical experiments in Appendix \ref{sec:appendix_3} for an illustration of both the improvement of the Anderson acceleration with respect to the convergence speed of the fixed point iteration and of the over-performance of Lloyd I method with respect to the stability of the Newton-Raphson algorithm.
Appendix \ref{sec:appendix_3} is by no means intended to be exhaustive, since it illustrates the performance of the Anderson acceleration algorithm in some examples. 

\subsection{The Large Time Step Algorithm}\label{subsec:fast_exponentiation}
The \LTSA is employed to recover the transition probability matrix associated to a time and space discretisation of a LV model. We start here by recalling some known results about Markov processes, that will be used in what follows.
We work under the following assumption:

\begin{assumption}
The asset price process $X$ follows the dynamics in Equation \eqref{eq:dynamic}, where the drift and diffusion coefficients $b$ and $\sigma$ are piecewise-constant functions of time.
\end{assumption}

Let us consider the Markov process $X$ in Equation \eqref{eq:dynamic} and let us denote by $p(t^{'},\gamma^{'} \vert t, \gamma)$, with $0 \le t < t' \le T$ and $\gamma, \gamma' \in \mathbb R$, the transition probability density from state $\gamma$ at time $t$ to state $\gamma'$ at time $t'$. Under some non stringent assumptions, it is known that $p$, as a function of the backward variables $t$ and $\gamma$, satisfies the backward Kolmogorov equation (see \cite{karatzas2012brownian,kijima1997markov}):
\begin{equation}\label{eq:kolmogorow}
  \begin{split}
    &\frac{\partial p}{\partial t}(t',\gamma'|t,\gamma)+(\mathcal{L}p)(t',\gamma'|t,\gamma)=0\quad\text{for}\quad (t,\gamma) \in (0,t') \times \mathbb R \,,\\
    &p(t,\gamma'|t,\gamma)=\delta(\gamma-\gamma^{'})\, ,
  \end{split}
\end{equation}
where $\delta$  is the Dirac delta and $\mathcal{L}$ is the infinitesimal operator associated with the SDE \eqref{eq:dynamic}, namely a second order differential operator acting on functions $f: \mathbb{R}_+ \times \mathbb{R} \rightarrow \mathbb{R}$ belonging to the class $C^{1,2}$ and defined as follows:
\begin{equation}\label{eq:markov_generator}
  (\mathcal{L} f)(t,\gamma)=b(t,\gamma)\frac{\partial f}{\partial \gamma}(t,\gamma)+\frac{1}{2}\sigma^2(t,\gamma) \frac{\partial^{2}f}{\partial \gamma^{2}} (t,\gamma)\,.
\end{equation}
The solution to Equation~\eqref{eq:kolmogorow} can be formally written as
\begin{equation}\label{eq:transDensity}
  p(t^{'},\gamma^{'}| t,\gamma)=\re^{(t^{'}-t)\mathcal{L}}p(t,\gamma).
\end{equation}

\medskip

The \LTSA consists in approximating the transition probabilities relative to a discrete-time finite-state Markov chain approximation of $X$ using Equation \eqref{eq:transDensity}. We report now a simple example to clarify how the \LTSA works.
\begin{ex}\label{ex:bsigma}
Consider for example the case when $b$ and $\sigma$ in Equation \eqref{eq:dynamic} are defined as:
\begin{equation*}
  \begin{split}
&b(t,X_{t})=b_{1}(X_{t})\indicator_{[0,T_1]} (t)+b_{2}(X_{t}) \indicator_{[T_{1}, T_{2}]}(t)\,,\\
&\sigma(t,X_{t})=\sigma_{1}(X_{t}) \indicator_{[0,T_1]} (t) + \sigma_{2}(X_{t}) \indicator_{[T_{1}, T_{2}]}(t) \,,
  \end{split}
\end{equation*}
where $T_{1}$ and $T_{2}=T$ are two target maturities and $b_1,b_2,\sigma_1,\sigma_2$ suitable functions. 
The transition probabilities in this case are explicitly given. In particular, if we denote by $\mathcal{L}_{\Gamma}^{1}$ and $\mathcal{L}_{\Gamma}^{2}$ the infinitesimal Markov generators of the Markov chain approximation of $X$ in $[0,T_{1}]$ and $[T_{1},T_{2}]$ respectively, the transition probabilities between any two arbitrary dates $t$ and $t'$ are given by:
\begin{equation*}\label{eq:receipe_tpd}
  \begin{split}
&\re^{(t' -t)\mathcal{L}_{\Gamma}^{1}}\text{\,\,\,\,\,\,\,\,\,\,\,\,\,\,\,\,\,\,\,\,\,\,\,\,\,\,\,\,\,\,\,\,\,\,\,\,for}\quad 0\leq t \leq t' \leq T_1\,,\\
&\re^{(T_1-t) \mathcal{L}_{\Gamma}^{1}}\re^{(t' -T_{1})\mathcal{L}_{\Gamma}^{2}}\qquad\text{for}\quad 0\leq t \leq T_1\leq t' \leq T_2\,,\\
&\re^{(t' - t) \mathcal{L}_{\Gamma}^{2}}\qquad\text{\,\,\,\,\,\,\,\,\,\,\,\,\,\,\,\,\,\,\,\,\,\,\,for}\quad T_1\leq t \leq t' \leq T_2\,.
  \end{split}
\end{equation*} 
In real market situations the above assumption on $b$ and on $\sigma$ is not at all restrictive, as we are going to see in Section \ref{sec:Numerical}.
\end{ex}

Let us now give more details on the algorithm. First of all, once a time discretisation grid $\{u_0, u_1, \dots, u_m \}$ has been chosen (think for example to the calibration pillars or to the expiry dates of the calibration dates of vanilla options), we need to obtain the space discretisation grids $\Gamma_k, 0 \le k \le m$.
Here these grids do not stem from the minimization of any distortion function, since they are defined quite flexibly as follows: 
$$
\Gamma_{0}\equiv x_0 \quad \text{and}\quad \Gamma_{k}\equiv \Gamma \doteq \{\gamma_{1},\dots,\gamma_{N}\},  \quad k = 1,\ldots,m. 
$$ 
This represents a major difference with respect to the \RMQA.

Then, the method consists in discretizing, opportunely, the Markov generator $\mathcal{L}$ and in calculating, in an effective and accurate way, the transition probabilities. As regards the discretisation, ~\cite{albanese2007convergence} gives a recipe to construct the discrete counterpart of $\mathcal{L}$ -- denoted by $\mathcal{L}_{\Gamma}$ -- so that the Markov chain approximation of $X$ converges to the continuous limit process in a weak or distributional sense~\cite{kushner2013numerical}. In particular, $\mathcal{L}_{\Gamma}$ corresponds to the discretisation of Equation~\eqref{eq:markov_generator} through an explicit Euler finite difference approximation of the derivatives~\cite{mitchell1980finite}. In Appendix \ref{sec:appendix_4} we provide more details on the discretisation of $\mathcal{L}$. 

Once $\mathcal{L}_{\Gamma}$ is constructed, one writes a (matrix) Kolmogorov equation for the transition probability matrix. In particular, using operator theory~\cite{albanese2007operator}, the transition probability matrix between any two arbitrary dates $u_k$ and $u_{k'}$ with $0 \le u_k<u_{k'} \le T$ can be expressed as a matrix exponential.

\begin{rem}
The piecewise time-homogeneous feature of the process $X$ plays a crucial role as regards the computational burden required to compute the transition probability matrix. Indeed, in case of time-dependent drift and volatility coefficient, it can no longer be expressed, in a straightforward way, as the exponential of the (time-dependent) Markov generator $\mathcal{L}_{\Gamma}$ (see, for instance ~\cite{blanes2009magnus}).
\end{rem}
The \LTSA is computationally convenient with respect to the \RMQA when pricing path-dependent derivatives whose payoff specification requires the observation of the asset price on a pre-specified set of dates, for example, $\{u_{0},u_{1},\dots,u_{m}\}$. Indeed, in this case we first calculate off-line the $m$ transition matrices as in Equation \eqref{eq:receipe_tpd}, then we price the derivative products via Monte Carlo with coarse-grained resolution.
In Figure~\ref{fig:LTSA} we plot an example of a possible path corresponding to the case $m=3$.
\begin{figure}[!ht]
  \centering
  \includegraphics[scale=0.3]{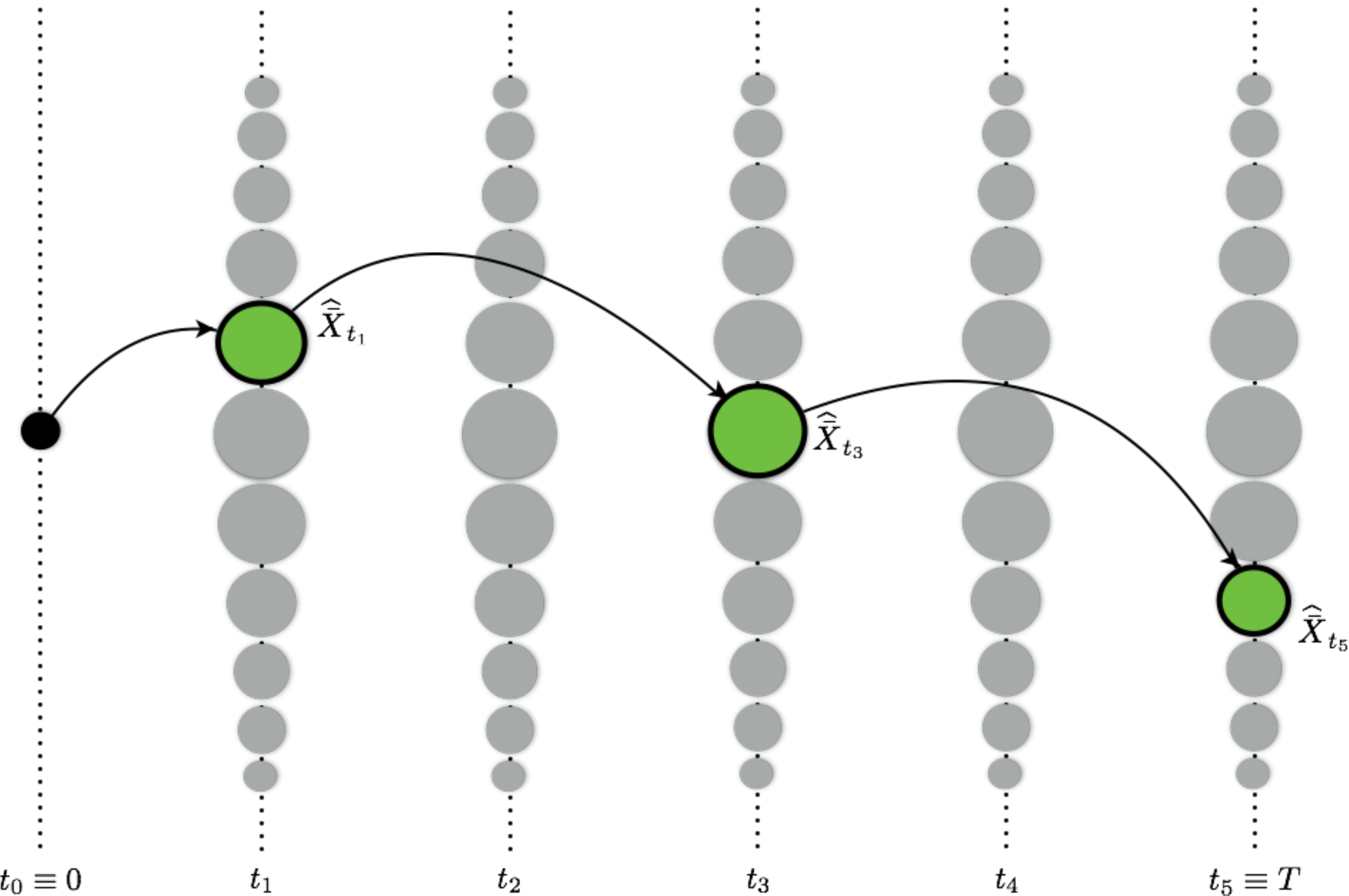}
  \caption{Example of one Monte Carlo path sampled with the \LTSA with $u_1=t_1$, $u_2=t_3$, and $u_3=t_5$.}
  \label{fig:LTSA}
\end{figure}
This major difference between \RMQA and \LTSA becomes more evident looking at Figures~\ref{fig:RMQA_transition} and~\ref{fig:RMQA_marginal}, where we plot, respectively, a Monte Carlo simulation connecting the initial point $x_0$ with a random final point $\widehat{\bar{X}}_{t_5}$, and a direct jump to date simulation to random points $\widehat{\bar{X}}_{t_k}$ with $k=1,\ldots,5$, respectively.
\begin{figure}[!ht]
  \centering
  \includegraphics[scale=0.3]{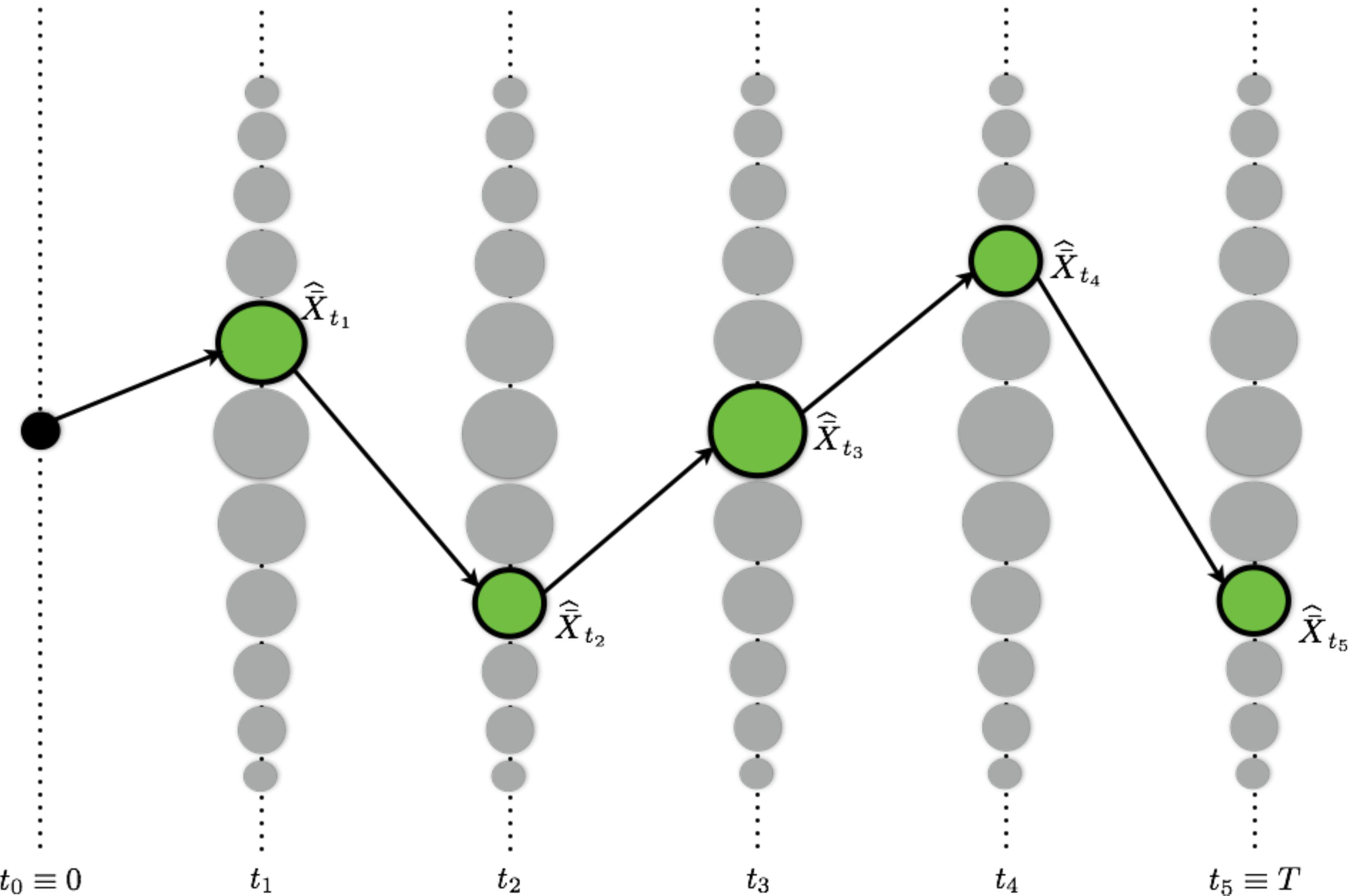}
  \caption{Example of one Monte Carlo path sampled with \RMQA over a time-grid computed with six time buckets.}
  \label{fig:RMQA_transition}
\end{figure}
\begin{figure}[!ht]
  \centering
  \includegraphics[scale=0.3]{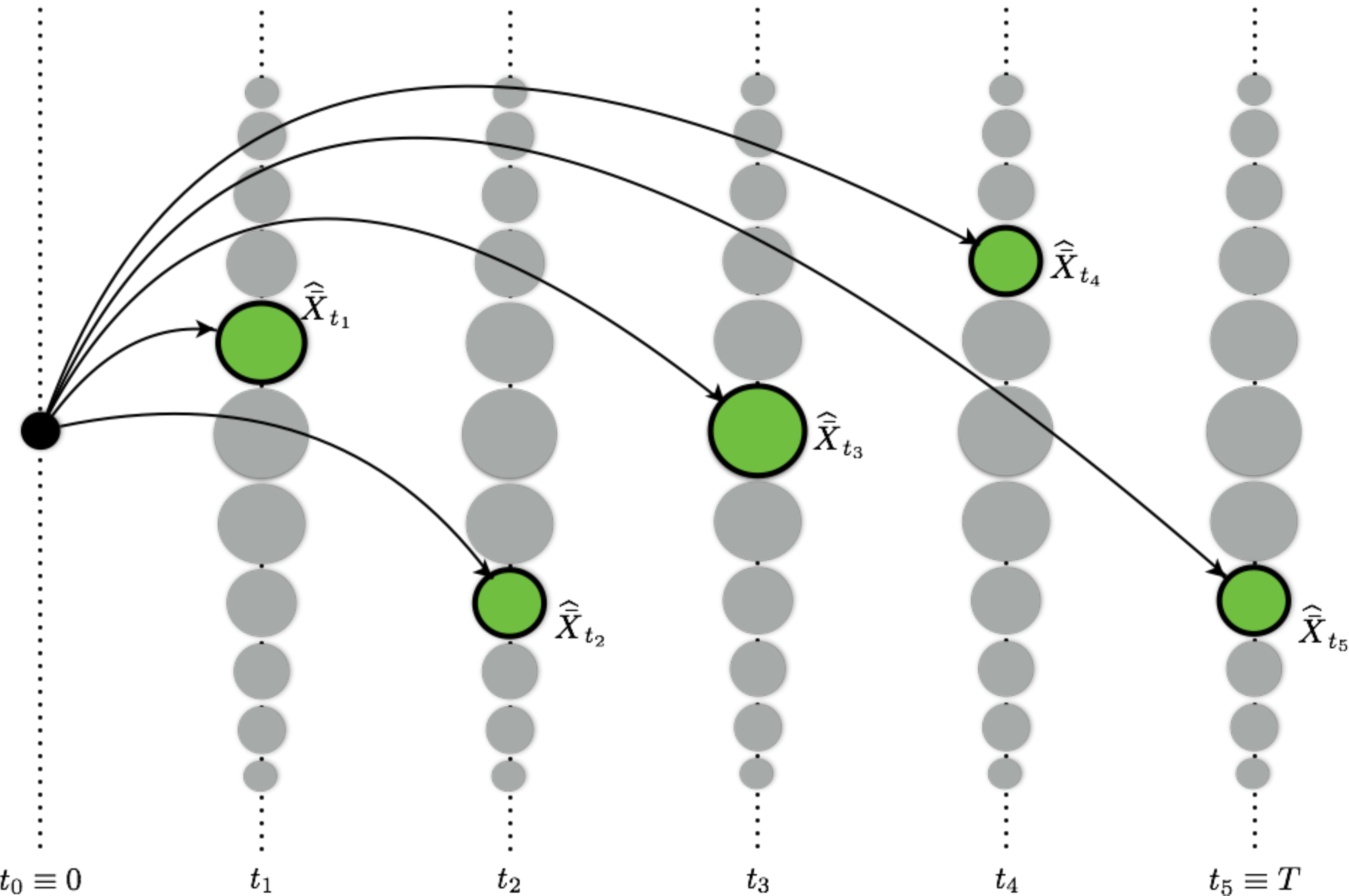}
  \caption{Example of direct transitions from the initial point $x_0$ to random points $\widehat{\bar{X}}_{t_k}$ with $k=1,\ldots,5$ computed with the \RMQA.}
  \label{fig:RMQA_marginal}
\end{figure}

We underline that, both the \RMQA and the \LTSA enable the computation of the price of vanilla options by means of a straightforward scalar product. Indeed, the price of a vanilla call option with strike $K$ and maturity $T = t_{m}$ can be computed as follows
\begin{equation*}
  \sum_{i=1}^{N}\mathbb{P}(\gamma_{i}^{m}| x_0)(\gamma_{i}^{m}-K)^{+}\,.
\end{equation*}

\subsubsection{LTSA implementation: more technical details}
In order to compute effectively and accurately a matrix exponential and so recover the transition probability for the \LTSA, we use the so called Scaling and Squaring method along with Pad\`e approximation. In particular, we implement the version of the method proposed by Higham in \cite{higham2005scaling} because it outperforms, both in efficiency and accuracy, previous implementations proposed in~\cite{sidje1998expokit} and \cite{ward1977numerical}. We now give a brief drawing of the algorithm implemented in \cite{higham2005scaling}, outlining its practical features but without giving all the mathematical details behind it. For a more extensive analysis on the method we refer to the original paper. Besides, we refer to~\cite{baker1996pade} for an extensive description of Pad\`e approximation. 

The Scaling and Squaring algorithm exploits the following trivial property of the exponential function:
\begin{equation}\label{eq:property}
  \re^{\widetilde{\mathcal{L}}_{\Gamma}}=\left(\re^{\frac{\widetilde{\mathcal{L}}_{\Gamma}}{\beta}}\right)^{\beta}\,,
\end{equation}
where $(t_{\tilde{k}}-t_{k})\mathcal{L}_{\Gamma} \doteq \widetilde{\mathcal{L}}_{\Gamma}$, together with the fact that $\re^{\widetilde{\mathcal{L}}_{\Gamma}}$ is well approximated by a Pad\`e approximation near the origin, that is for small $\|\widetilde{\mathcal{L}}_{\Gamma}\|$, where $\|\cdot\|$ is any subordinate matrix norm. In particular, Pad\`e approximation estimates $\re^{\widetilde{\mathcal{L}}_{\Gamma}}$ with the ratio between two (matrix) polynomials of degree $13$. The mathematical elegance of the Pad\`e approximation is enhanced by the fact that the two approximating polynomial are known explicitly. 

\noindent The main hint of the Scaling and Squaring method is to choose $\beta$ in Equation \eqref{eq:property} as an integer power of 2, $\beta=2^{n}$ say, so that $\widetilde{\mathcal{L}}_{\Gamma}/\beta$ has norm of order 1, to approximate $\re^{\widetilde{\mathcal{L}}_{\Gamma}/\beta}$ by a Pad\`e approximation, and then to compute $\re^{\widetilde{\mathcal{L}}_{\Gamma}}$ by repeating squaring $n$ times. In particular, we define $\delta t \doteq (t_{k'}-t_{k})/2^{n}$. 
\noindent We use Pad\`e approximation because of the usage of the explicit Euler Scheme for the discretisation of $\mathcal{L}$ (see Appendix \ref{sec:appendix_4}). Indeed, one needs to impose the so called Courant condition for the matrix $\delta t \mathcal{L}_{\Gamma}$. The Courant condition requires that $\|\delta \mathcal{L}_{\Gamma}\|_{\infty}<1$. This translates into the following stringent condition for $\delta t$
\begin{equation*}
  \delta t < \frac{1}{2}\frac{(\Delta \gamma)^2}{\sigma(\gamma_{i})},\,\, \forall\,1\le i \le N\,.
\end{equation*}
The usage of the Pad\`e approximation permits to relax the last constraint. In particular, the implementation in~\cite{higham2005scaling} allows $\|\delta \mathcal{L}_{\Gamma}\|_{\infty}$ to be much larger.

\section{Financial applications}\label{sec:Numerical}
In this final section we present and discuss how results achieved in the previous Sections~\ref{sec:Monte_Carlo} and~\ref{sec:transition} can be applied to finance, and in particular to option pricing in the FX market, where spot and forward contracts, along with vanilla and exotic options are traded (see, for instance ~\cite{reiswich2012fx} for a broad overview on FX market). In particular, we consider the following types of path-dependent options: (i) Asian calls, (ii) up-and-out barrier calls, (iii) automatic callable (or auto-callabe).
We choose two different models for the underlying EUR/USD FX rate: a LV model as a benchmark  and the Constant Elasticity of Variance model (henceforth CEV)~\cite{cox1975notes}, coming from the academic literature.

\subsection{Model and payoff specifications}\label{sub:Specs}

Let us first introduce some notations relative to the LV model and to the CEV model. Recall that we have assumed in Section \ref{sec:Monte_Carlo} deterministic interest rates. Moreover, we indicate by $X_{t}$ the spot price at time $t$ of one EUR expressed in USD and by $X_{t}(T)$ the corresponding forward price at time $t$ for delivery at time $T$. 
\noindent We introduce the so-called normalized spot exchange rate $x_{t}^{LV}\doteq X_{t}/X_{0}(t)$, all $t\in [0,T]$ (where the superscript ``LV'' clearly stands for Local Volatility)
 and we suppose that the process $x^{LV} \doteq (x_{t}^{LV})_{t\in [0,T]}$ follows the SDE
\begin{equation*}
  \begin{cases}
    \rd x_{t}^{LV} = x_{t}^{LV} \eta(t,x_{t}^{LV}) \,\rd W_{t}\,,\\
    x_{0}^{LV}=1\,.
  \end{cases}
\end{equation*}%
Hence, in this LV model, $x^{LV}$ corresponds to the underlying process $X$ introduced in Section \ref{sec:Monte_Carlo}, and besides making reference to Equation \eqref{eq:dynamic} we have  $b(t,x_{t}^{LV})=0$ and $\sigma(t, x_{t}^{LV})=\eta(t,x_{t}^{LV})x_{t}^{LV}$, where $\eta:[0,T]\times \mathbb{R}\rightarrow \mathbb{R}_{+}$ corresponds to the local volatility function. Specifically, it is a cubic monotone spline for fixed $t\in [0,T]$ (see~\cite{fritsch1980monotone} for an overview on interpolation technique) with flat extrapolation. The set of points to be interpolated is determined numerically during the calibration procedure\footnote{We use the calibration procedure proposed in~\cite{reghai2012local} and refined in~\cite{pallavicini2015}. This procedure is particularly robust. Indeed, the resulting local volatility surface is ensured to be a smooth function of the spot. The data set is available upon requests.}. In particular, this procedure leads to a piecewise time-homogeneous dynamics for the process $x^{LV}$.

\noindent As a second example we consider the CEV, i.e., we assume that the asset price process $X$ follows a CEV dynamics
\begin{equation*}
  \begin{cases}
    \rd X_{t}=r X_{t} \,\rd t+\sigma X_{t}^{\alpha} \,\rd W_{t} \,,\\
    X_{0}=x_{0}\in \mathbb{R}_{+}\,,
  \end{cases}
\end{equation*}%
where $r\in \mathbb{R}_{+}$ is the risk-free interest rate, $\sigma \in \mathbb{R}_{+}$ is the volatility, and $\alpha > 0$ is a constant parameter.\\

Then, given a time discretisation grid $\{ 0 =t_0, t_1, \dots, t_n=T \}$ on $[0,T]$ as in Section  \ref{sec:Monte_Carlo} and making reference to the Euler-Maruyama scheme in Equation \eqref{eq:Euler_Maruyama} we consider the unidimensional payoff specifications below. In particular, we compute the price at time $t_{0} = 0$.
\begin{itemize}
\item[i)] \textbf{Asian calls.} The discounted payoff function of a discretely monitored Asian call option is
\begin{equation*}\label{eq:Asian_option}
  \begin{split}
    \payoff_\text{A}(\bar{X}_{0},\dots,\bar{X}_{t_{n}})& \doteq e^{-r(t_{n}-t_{0})}\max\left(\frac{1}{n+1}\sum_{i=0}^{n}\bar{X}_{t_{i}}-K,0\right),
  \end{split}
\end{equation*}
where $K$ is the strike price and $T>0$ the maturity. 
\item[ii)] \textbf{Up-and-out barrier calls.} We consider barrier options of European style.  The discounted payoff at maturity $T>0$ of an up-and-out barrier call is given by:
\begin{equation}\label{eq:Barrier}
\re^{-r(t_{n}-t_{0})}\max(X_{T}-K,0)\indicator_{\{\tau > T\}}
\end{equation}%
where $K$ is the strike price, $\tau\doteq  \inf\{t \ge 0: X_{t}\ge B\}$ and $B$ is the upper barrier. It is known -- see for instance \cite{lapeyre2003understanding} -- that, whenever we discretise the continuous time discounted payoff in Equation \eqref{eq:Barrier} by defining 
\begin{equation*}
  \payoff_\text{B}(\bar{X}_{t_{0}},\dots,\bar{X}_{t_{n}}) \doteq \re^{-r(t_{n}-t_{0})}\max(\bar{X}_{t_{n}}-K,0)\prod_{k=0}^{n}\indicator_{\{\bar{X}_{t_k}<B\}}\,,
\end{equation*}
we overestimate the price of the option. Actually, we do not take into account the possibility that the asset price could have crossed the barrier for some $t \in (t_{k},t_{k+1})$, $0\le k \le n-1$. In \cite{glasserman2003monte,lapeyre2003understanding} the authors propose a strategy to obtain a better approximation of the price of the option in Equation \eqref{eq:Barrier} when employing MC simulation. It consists in checking, at each time step $t_{k}=k\Delta t$, $0\le k \le n-1$, and for all the MC paths $l$, $1\le l \le N_{MC}$,  whether the simulated path $\bar{X}_{t_{k}}^{(l)}$ has reached the barrier $B$ or not. So, first one computes the probability
\begin{equation*}
  p_{k}^{(l)}\doteq 1-\exp\left[-\frac{2}{\sigma_{k}\Delta t}(B-\bar{X}_{t_{k}}^{(l)})(B-\bar{X}_{t_{k+1}}^{(l)})\right],
\end{equation*}
with $\sigma_{k}$ the diffusive coefficient of the underlying asset price in $(t_k,t_{k+1})$, then one simulates a random variable from a Bernoulli distribution with parameter $1-p_{k}^{(l)}$: if the outcome is favourable the barrier has been reached in the interval $(t_{k},t_{k+1})$ and the price associated to the $l$-th path is zero. Otherwise, the simulation is carried on to the step further. Consistently, the adjusted discounted payoff for a discretely monitored up-and-out call barrier option reads:
  \begin{equation*}
    \re^{-r(t_{n}-t_{0})}\max(\bar{X}_{t_{n}}-K,0)\prod_{k=0}^{n-1}p_{k}\,.
  \end{equation*}
\item[iii)] \textbf{Automatic callable (or auto-callable)\footnote{They were first issued in the U.S. by BNP Paribas in August 2003 as cited for instance in ~\cite{deng2011modeling}.}.} The discounted payoff of an auto-callable option with unitary notional is given by 
\begin{equation*}
  \begin{split}
    &\begin{cases}
      \re^{-r(t_{i}^{c}-t_{0})}Q_{i}\qquad\text{\,\,if } \bar{X}_{t_{j}^{c}}< X_{0} \mathrm{b} \le \bar{X}_{t_{i}^{c}} \qquad\text{for all } j<i\,,\\
      \re^{-r(t_{n}-t_{0})}\frac{X_{t_{n}}}{X_{0}}\qquad\text{if } \bar{X}_{t_{i}^{c}}<X_{0} \mathrm{b} \qquad\text{\,\,\,\,\,\,\,\,\,\,\,\,\,\,\,\,\,for all }i=1,\dots,m\,,\\
    \end{cases}
  \end{split}
\end{equation*}
where $\{t_1^{c}, \dots, t_{m}^{c}\}$ is a set of pre-fixed call dates, $\mathrm{b}>X_{0}$ is a pre-fixed barrier level, and $\{Q_{1}, \dots, Q_{m}\}$ is a set of pre-fixed coupons.  The set of call dates $\left\{t_{1}^{c}, \dots, t_{m}^{c}\right\}$ does not coincide with the set of times of the Euler scheme discretisation $\left\{t_{0}, \dots, t_{n}\right\}$. In particular, the latter has finer time resolution grid.
\end{itemize}
\medskip
We show in Section \ref{sub:numerical_discussion} that all previous payoffs can be priced efficiently by using our novel algorithm, i.e., by reverting the Monte Carlo paths and simulating them from maturity back to the initial date.

\subsection{Numerical results and discussion}\label{sub:numerical_discussion}
Let us introduce some terminology that we will use in the summary Tables of our numerical results. In particular, we will termed: (i) \textit{Euler Scheme} prices obtained via a Monte Carlo procedure on the process $\bar{X}$, (ii) \textit{Forward} prices obtained via a forward Monte Carlo procedure on $\widehat{\bar{X}}$ from the starting date to the maturity, (iii) $\textit{Backward}$ prices computed through the Backward Monte Carlo algorithm, and finally, (iv) the \textit{Benchmark} price is an \textit{Euler Scheme} price (in case of Asian call and up-and-out barrier call options) or a \textit{Forward} price (in case of auto-callable option) whose estimation error is negligible respect to the significant digits reported. Besides, in brackets we will report the numerical estimation error corresponding to one standard deviation.\\
Let us now stress some aspects related to the implementation of the backward Monte Carlo algorithm along with the procedures described in Sections~\ref{subsec:rmq_algorithm} and \ref{subsec:fast_exponentiation}.\\
In order to have a meaningful comparison between \textit{Euler Scheme} and \textit{Backward} prices and between \textit{Forward} and \textit{Backward} prices, for each of the $N^{+}$ points $\tilde{\gamma}_{i_{n}}^{n} \in \widetilde{\Gamma}_{n}$ we generate $N_{MC}^{i_{n}}$ random paths in such a way that $N_{MC}^{i_{n}} \times N^{+} = N_{MC}$, where $N_{MC}$ indicates the number of simulations employed to compute either \textit{Euler Scheme} or \textit{Forward} prices. The choice of the final domain of integration $\widetilde{\Gamma}_{n}$ depends on the payoff specification. In particular, $\widetilde{\Gamma}_{n} \doteq \{\gamma_{i}^{n} \in \Gamma_{n}:\text{ }K \le \gamma_{i}^{n} \le B \}$ when pricing up-and-out call barrier options and $\widetilde{\Gamma}_{n} = \Gamma_{n}$ when pricing In-The-Money (ITM), At-The-Money (ATM), Out-The-Money (OTM) Asian call options and auto-callable options.
Concerning the granularity of the state-space discretisation we fix the cardinality of the quantizers $\Gamma_{k}$, $1\le k \le n$, to 100. With this value the error on vanilla call option  prices, computed as
\begin{equation}\label{eq:error_2}
  |\sigma^{mkt}-\sigma^{alg}|
\end{equation}
is less than or equal to five basis point (recall that $1$ bp $=10^{-4}$), where in Equation~\eqref{eq:error_2}, $\sigma^{mkt}$ is the market implied volatility, whereas $\sigma^{alg}$ is the implied volatility computed by the backward Monte Carlo algorithm.\\
\noindent The stopping criteria for the \RMQA corresponds to $\|\Gamma_{k}^{l+1}-\Gamma_{k}^{l}\|\le 10^{-5}$, $1\le k \le n$, where $\Gamma_{k}^{l}$ is the quantizer computed by the algorithm at time $t_{k}\in \{t_{1},\dots,t_{n}\}$ at the $l$-th iteration. Moreover, in the Backward Monte Carlo algorithm case, for each point in $\widetilde{\Gamma}_{n}$ we generate $N_{MC}^{i_{n}} = N_{MC} \div \vert N^{+} \vert = 10^{4}\div \vert \widetilde{\Gamma}_{n} \vert$ random paths. Let us now come to the discussion of the numerical results.\\

\medskip

\noindent In Table \ref{tab:Barrier}  we report up-and-out barrier call option prices for both LV and CEV, as well as their relative estimation errors.  In order to test the performances of our algorithm we price ITM, ATM and OTM options. In particular, for both models the initial spot price is $\bar{X}_0 = 1.36$. This value corresponds to the value of the EUR/USD exchange rate at pricing date (23-June-2014). Besides, for both dynamics the value of the pair strike-barrier, $(K, B)$, is set to $(1.35, 1.39)$, $(1.36, 1.39)$ and to $(1.37, 1.39)$ for ITM, ATM and OTM up-and-out call barrier options respectively. The maturity $T$ is $6$ months and the number of Euler steps is $n = 51$. For CEV model we fix $\alpha = 0.5$ and $r = 0.32 \%$ (the latter corresponds to the value of the 6 months domestic interest rates implied by the forward USD curve at pricing date). Instead, as regards the parameter $\sigma$ we vary it from $\sigma = 5 \%$ to $\sigma = 20 \%$ with steps $\Delta \sigma = 5\%$. \\
\noindent Panel A of Table \ref{tab:Barrier} compares the efficiency of the Euler Scheme Monte Carlo with that of the Backward Monte Carlo for the Local Volatility dynamics. Panel B, instead, compares the efficiency of the two algorithms for the CEV dynamics. 
For both models the Backward Monte Carlo algorithm exhibits better performances than the Euler Scheme Monte Carlo. 
More precisely, for LV the ratio between the estimation error of the Euler Scheme MC and that of the Backward MC, henceforth $Error$ $ratio$, is $2.2$, $2.5$ and $3.1$ for ITM, ATM and OTM options respectively. As regards the CEV model, Figure \ref{fig:barrier} summarizes the results. In particular, gain in efficiency is more evident if we increase the value of the parameter $\sigma$. Intuitively, this happens because the probability for the price paths to hit the barrier $B$ over the life of the option increases with the increasing of $\sigma$. Moreover, for a fixed value of  $\sigma$ gain in efficiency is more evident when pricing OTM options. This happens because for OTM options a relevant number of forward paths do not contribute to the payoff and, in order to increase the pricing accuracy of the Euler Scheme MC, it would be necessary to force paths to sample the region in which the payoff is different from zero, namely between the strike $K$ and the barrier $B$.

\begin{table}[!h]\centering
  \ra{1}
\begin{tabular}{@{}lccc@{}}\toprule
  \multicolumn{4}{c}{\textbf{Up-and-out barrier call}}\\
  \cmidrule{1-4}
   \textbf{Algorithm}  & \textbf{ITM} & \textbf{ATM} & \textbf{OTM}  \\
  \cmidrule{1-4}
  \textbf{Panel A} &\multicolumn{3}{c}{\textbf{Local Volatility model}}\\
  \cmidrule{1-4}
 \textit{Euler Scheme}	 &  1.063E-3 (3.8E-5)		&	4.54E-4 (2.2E-5)	&  1.41E-4( 1E-5)\\
 	\textit{Backward} 	 &  1.064E-3 (1.7E-5)		& 	4.67E-4 (  9E-6)	&  1.45E-4( 3E-5)\\
	\textit{Benchmark}  	 &  1.055E-3   		            &   4.58E-4				&  1.44E-4			\\  
  \cmidrule{1-4}
\textbf{Panel B} &\multicolumn{3}{c}{\textbf{CEV model}}\\
  \cmidrule{1-4}
  \textbf{$\sigma = 5\%$}         &			&		&		\\
  \cmidrule{1-4}
  	\textit{Euler Scheme}	 & 2.431E-3 (6.7E-5)		&	1.133E-3 (4.1E-5)	&  3.11E-4 (1.7E-5)			\\
 	\textit{Backward} 	     & 2.500E-3 (3.1E-5)		& 	1.116E-3 (1.6E-5)	&  3.49E-4 (6E-6)			\\
	\textit{Benchmark}  	    	 & 2.501E-3	   		    &   1.073E-3			    &  3.49E-4					\\  
  \cmidrule{1-4}
  \textbf{$\sigma = 10\%$}      		 &						&						&					\\
  \cmidrule{1-4}
  \textit{Euler Scheme}		&  4.53E-4 (3.0E-5)		&	1.86E-4 (1.8E-5)	&  5.31E-5 (7.5E-6) 		\\
  \textit{Backward}        	&  3.94E-4 (1.0E-5)		&   1.69E-4 (  5E-6)	&  5.39E-5 (1.8E-6)		\\
  \textit{Benchmark}    		&  4.03E-4  		        & 	1.70E-4			&  5.49E-5					\\  
  \cmidrule{1-4}
  \textbf{$\sigma = 15\%$}       		&						&		&		\\
  \cmidrule{1-4}
  \textit{Euler Scheme}		&  1.32E-4 (1.6E-5)		& 	5.59E-5 (9.1E-6)	&  1.48E-5 (4.0E-6)		\\
  \textit{Backward}      	&  1.28E-4 (5  E-6)		&   5.57E-5 (2.3E-6)	&  1.64E-5 (  8E-7)    	\\
  \textit{Benchmark}    		&  1.19E-4  				&   5.56E-5			&  1.64E-5 				\\  
  \cmidrule{1-4}
  \textbf{$\sigma = 20\%$}       		&							&		&		\\
  \cmidrule{1-4}
  \textit{Euler Scheme}   	 & 4.85E-5 (9.3E-6)		& 2.91E-5 (6.8E-6)    &  6.2E-7 ( 4.4E-7)	\\
  \textit{Backward}     		 & 5.80E-5 (2.6E-6)		& 2.53E-5 (1.2E-6)    &  8.1E-7 ( 5E-8)							\\
  \textit{Benchmark}   		 & 5.52E-5  				& 2.43E-5 			 &   7.5E-7					\\  
  \bottomrule
\end{tabular}
\caption{Numerical values for \textit{Euler Scheme} and \textit{Backward} prices for an up-and-out barrier call option for both LV and CEV model. \textit{Errors} correspond to one standard deviation. The initial spot price is $\bar{X}_{0} = 1.36$, whereas the pair strike-barrier is set to $(1.35, 1.39)$, $(1.36, 1.39)$, $(1.37, 1.39)$ for ITM, ATM and OTM options respectively.
}
\label{tab:Barrier}
\end{table}
\begin{figure}[!h]
\begin{center}
 \includegraphics[scale=0.75]{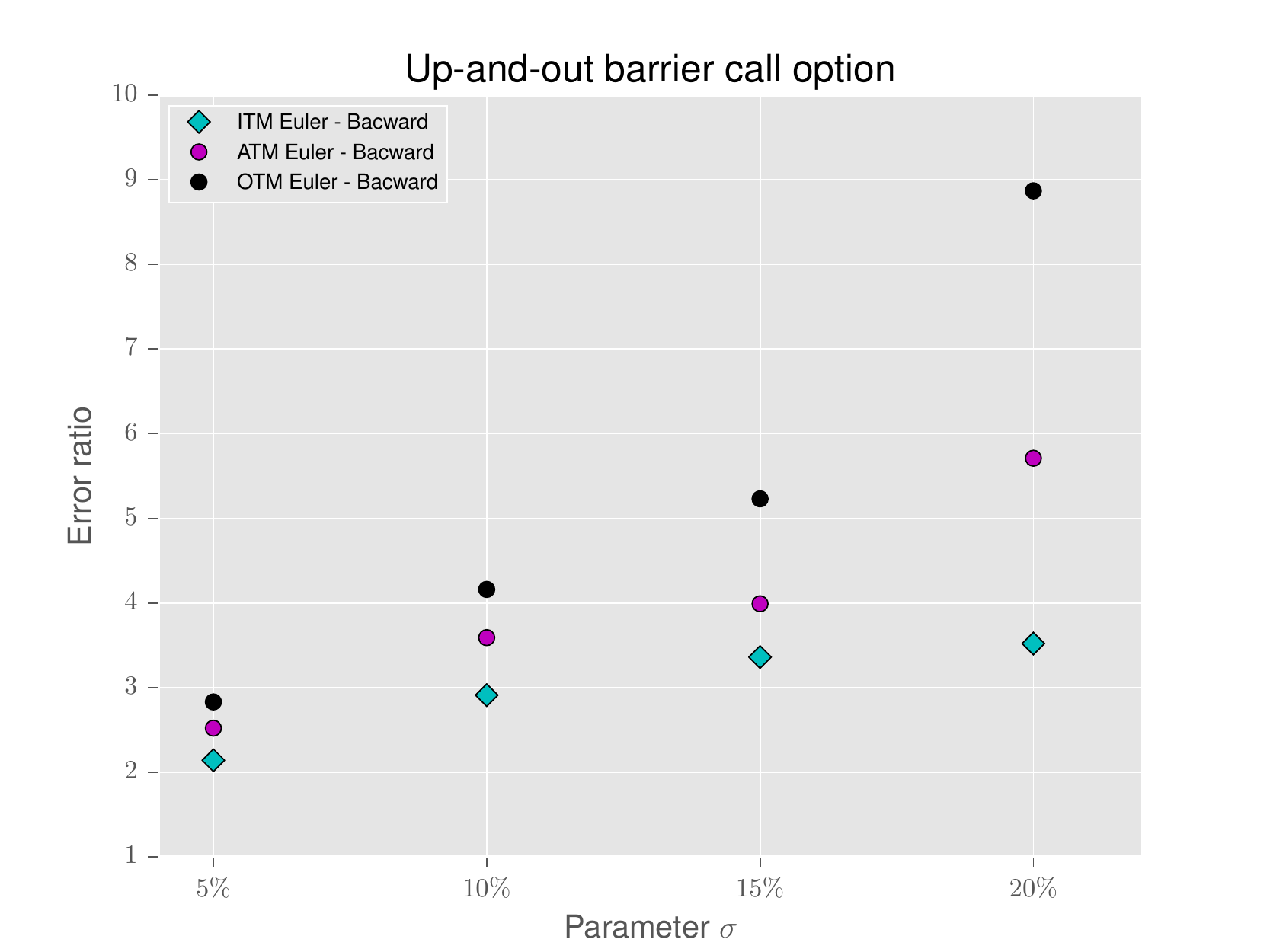} 
 \caption{Plot of the $Error$ $ratio$ as a function of the parameter $\sigma$ for CEV model when pricing up-and-out barrier call option. The initial spot price is $\bar{X}_{0} = 1.36$, whereas the pair strike-barrier is set to $(1.35, 1.39)$, $(1.36, 1.39)$, $(1.37, 1.39)$ for ITM, ATM and OTM options respectively.}
\label{fig:barrier} 
 \end{center}
\end{figure}

\medskip
In Table \ref{tab:Asian} we compare performances of the Euler Scheme Monte Carlo with those of the Backward Monte Carlo when pricing Asian call options, for both LV and CEV model.  We set $\bar{X}_{0} = 1.36$. As done for up-and-out barrier call we test the efficiency of our novel algorithm when pricing ITM ($K = 1.35$), ATM ($K=1.36$) and OTM ($K=1.37$) options. The maturity $T$ is 6 months and $n=51$. Also in this case, for CEV model we fix the value of the risk-free rate $r = 0.32\%$ and of $\alpha = 0.5$, and we vary the value of $\sigma$ from $5\%$ to $20\%$ with steps $\Delta \sigma = 5\%$.
Results for the LV dynamics are reported in Panel A of Table \ref{tab:Asian}, whereas Panel B reports the results for the CEV. Table \ref{tab:Asian} suggests that the strategy of reverting the MC paths and simulating them from maturity back to starting date is an effective alternative to Euler MC also for Asian call options. In this case the improvement in efficiency derives from the fact that with Backward MC we decide the number of paths to sample from each of the final points in $\Gamma_{n}$, sampling efficiently also those regions that are infrequently explored by the price process because of its diffusive behaviour.
Figure \ref{fig:asian} suggests that the importance of this feature is more evident when pricing OTM options. Besides, the $Error$ $ratio$ is almost constant across the value of $\sigma$ for a fixed scenario (ITM, ATM or OTM).

\begin{table}[!h]\centering
  \ra{1}
\begin{tabular}{@{}lccc@{}}\toprule
  \multicolumn{4}{c}{\textbf{Asian call}}\\
  \cmidrule{1-4}
   \textbf{Algorithm}  & \textbf{ITM} & \textbf{ATM} & \textbf{OTM}  \\
  \cmidrule{1-4}
  \textbf{Panel A} &\multicolumn{3}{c}{\textbf{Local Volatility model}}\\
  \cmidrule{1-4}
    \textit{Euler Scheme}& 0.013628 (0.000161)		&	0.009444 (0.000142)	&  0.006194 (0.000142)		\\
 	\textit{Backward}    & 0.013582 (0.000107)		& 	0.009384 (0.000092)	&  0.006124 (0.000092)		\\
	\textit{Benchmark}   & 0.013590	   		    		&   0.009398				&  0.006170						\\  
  \cmidrule{1-4}
\textbf{Panel B} &\multicolumn{3}{c}{\textbf{CEV model}}\\
  \cmidrule{1-4}
  \textbf{$\sigma = 5\%$}         &			&		&		\\
  \cmidrule{1-4}
  	\textit{Euler Scheme}	 & 0.015964 (0.000174)		&	0.009851 (0.000143)	&  0.005826 (0.000109)		\\
 	\textit{Backward} 	     & 0.015904 (0.000105)		& 	0.010014 (0.000085)	&  0.005565 (0.000065)		\\
	\textit{Benchmark}  	    	 & 0.015989	   		        &   0.010038				&  0.005634				\\  
  \cmidrule{1-4}
  \textbf{$\sigma = 10\%$}      		 &				    &						&						\\
  \cmidrule{1-4}
  \textit{Euler Scheme}		&  0.024682 (0.000321)	&	    0.019681 (0.000288)	&  0.014780 (0.000251) 				\\
  \textit{Backward}        	&  0.024709 (0.000190)	 &      0.019164 (0.000171)	&  0.015021 (0.000150)				\\
  \textit{Benchmark}    		&  0.024927  		         & 	0.019448				&  0.014921						\\  
  \cmidrule{1-4}
  \textbf{$\sigma = 15\%$}  &						&		&		\\
  \cmidrule{1-4}
  \textit{Euler Scheme}		&  0.033764  (0.00056)	& 0.028335 (0.000424)	&  0.024203 (0.000391)		\\
  \textit{Backward}      	&  0.034160  (0.00033)	& 0.028896 (0.000232)	&  0.024132 (0.000236)    	\\
  \textit{Benchmark}    		&  0.033991  			& 0.028867			    &  0.024011					\\  
  \cmidrule{1-4}
  \textbf{$\sigma = 20\%$}       		&							&		&		\\
  \cmidrule{1-4}
  \textit{Euler Scheme}   	 & 0.043553 (0.000604)     & 0.037880 (0.000554) &	0.033989  (0.000542)		\\
  \textit{Backward}     		 & 0.043117 (0.000363)	  & 0.037653  (0.000341) &  	0.033234  (0.000317)		\\
  \textit{Benchmark}   		 & 0.043276  			&   0.033366 			& 	0.033367					\\  
  \bottomrule
\end{tabular}
\caption{Numerical values for \textit{Euler Scheme} and \textit{Backward} prices for an Asian call option for both LV and CEV model. \textit{Errors} correspond to one standard deviation.
}
\label{tab:Asian}
\end{table}

\begin{figure}[!h]
\begin{center}
 \includegraphics[scale=0.75]{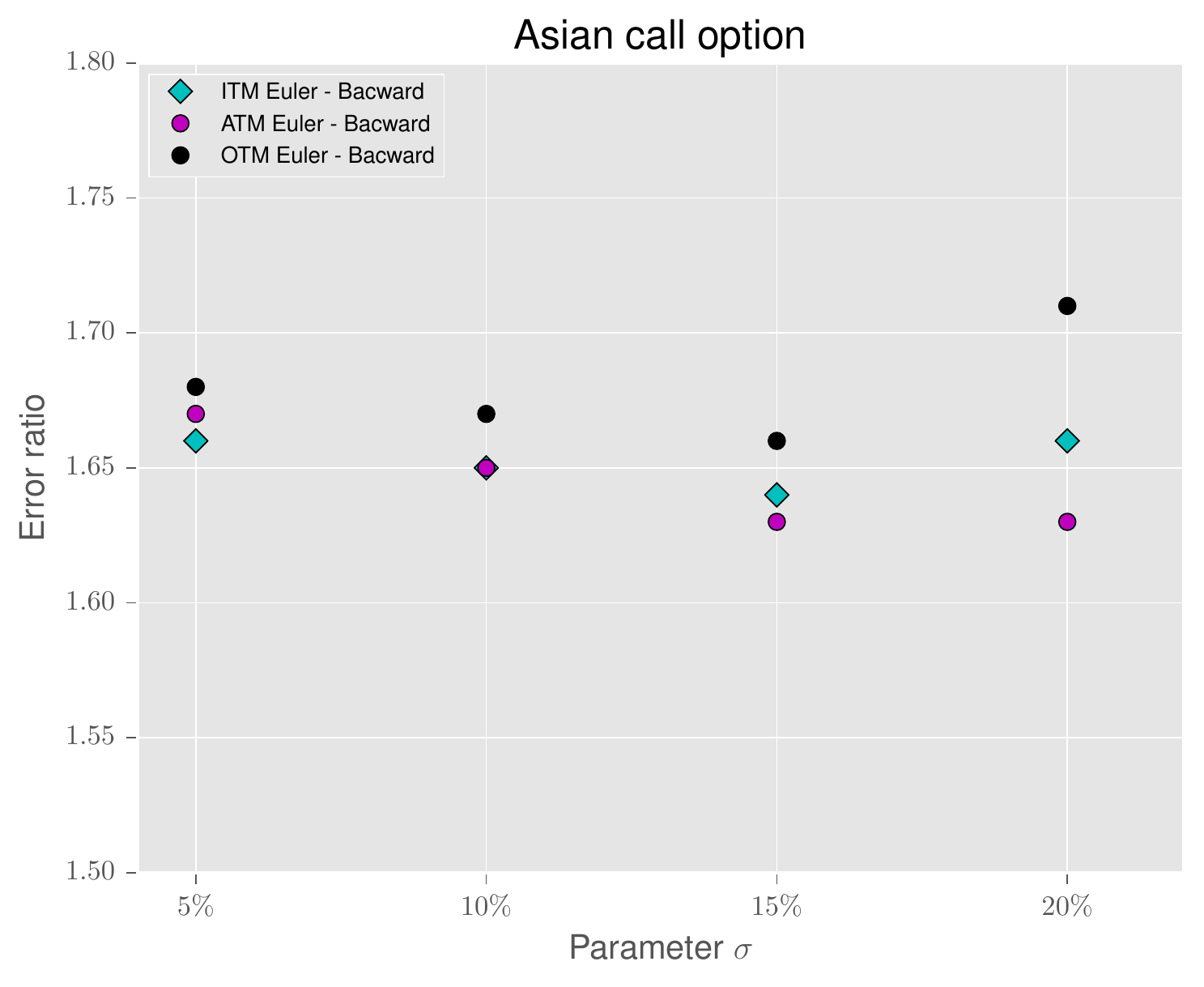} 
 \caption{Plot of the $Error$ $ratio$ as a function of the parameter $\sigma$ for CEV model when pricing Asian call option. The initial spot price is $\bar{X}_{0} = 1.36$, whereas the strike is set to 1.35, 1.36, 1.37 for ITM, ATM and OTM options respectively.}
\label{fig:asian} 
 \end{center}
\end{figure}

\medskip
In Table \ref{tab:Autocallable} we report prices of an auto-callable option for LV. In this case, we compare the Backward MC with the Forward MC. The Euler Scheme MC is ineffective for payoffs specifications which depend on the observation of the underlying on a pre-specified set of dates (such as auto-callable and European). The multinomial tree and the transition probability matrices are recovered by means of the \LTSA. In order to compare the two MC methodologies we fix a set of call-dates  $\{t_{1}^{c}, \dots, t_{4}^{c}\}$ and a set of pre-fixed coupon $\{Q_{1}, \dots, Q_{4}\}$ and we vary the value of the barrier b. Precisely, $\{t_{1}, \dots, t_{4}\} =\{1, 3, 6, 12 \}$ months, $\{Q_{1}, \dots, Q_{4}\} = \{5\%, 10\%, 15\%, 20\%\}$ with unitary notional, and $\text{b} \in \{\bar{X}_{0}, 1.05\bar{X}_{0}, 1.1\bar{X}_{0}\}$. As usual, at pricing date the EUR/USD exchange rate is $\bar{X}_{0}=1.36$.
Results in Table \ref{tab:Autocallable} show that the improvement on efficiency linked to the simulation of the paths backward from maturity to starting date increases with the increase of the value of the barrier $\text{b}$ (in real market usually one has $\text{b}>\bar{X}_{0}$). In particular, the ratio between the estimation error of the Forward MC and that of the Backward MC is $\approx 0.8$, $\approx 2 $ and $\approx 5.5 $ for $\text{b} = \bar{X}_{0}$, $\text{b} = 1.05 \bar{X}_{0}$ and $\text{b} = 1.1\bar{X}_{0}$ respectively. Intuitively, this happens because an increase in the value of the barrier $\text{b}$ makes the early exercise of the option less probable. In particular, a larger number of paths will reach the final domain of integration. So, thanks to our method we can sample with a predetermined number of paths regions in this domain that are infrequently explored by the price process.

\begin{table}[!h]\centering
  \ra{1}
\begin{tabular}{@{}lccc@{}}\toprule
  \multicolumn{4}{c}{\textbf{Auto-callable options}}\\
  \cmidrule{1-4}
   \textbf{Barrier b}  & \textbf{$\text{b}=\bar{X}_{0}$} & \textbf{$\text{b}=1.05\bar{X}_{0}$} & \textbf{$\text{b}=1.1\bar{X}_{0}$}  \\
  \cmidrule{1-4}
  \textbf{\textbf{Algorithm} } &\multicolumn{3}{c}{\textbf{Local Volatility model}}\\
  \cmidrule{1-4}
\textit{Forward}	 	& 0.04107 (0.00056)		&	0.01902 (0.00074)	&  0.00447 (0.00058)				\\
\textit{Backward} 	& 0.04099 (0.00072)		& 	0.01856 (0.00039)	&  0.00377 (0.00011)				\\
\textit{Benchmark}  	& 0.04099	   		    &   0.01820				&  0.00357						\\  
  \bottomrule
\end{tabular}
\caption{Numerical values for \textit{Forward} and \textit{Backward} prices for an auto-callable option for LV model. \textit{Errors} correspond to one standard deviation.
}
\label{tab:Autocallable}
\end{table}

\section{Conclusion}\label{sec:conclusion}

In this paper, we present a novel approach -- termed backward Monte Carlo -- to the Monte Carlo simulation of continuous time diffusion processes. We exploit recent advances in the quantization of diffusion processes to approximate the continuous process with a discrete-time Markov Chain defined on a finite grid of points. Specifically, we consider the Recursive Marginal Quantization Algorithm and as a first contribution we investigate a fixed-point scheme -- termed Lloyd I method with Anderson acceleration -- to compute the optimal grid in a robust way. As a complementary approach, we consider the grid associated with the explicit scheme approximation of the Markov generator of a piecewise constant volatility process. The latter approach -- termed Large Time Step Algorithm -- turns out to be competitive in pricing payoff specifications which require the observation of the price process over a finite number of pre-specified dates. Both methods -- quantization and the explicit scheme -- provide us with the marginal and transition probabilities associated with the points of the approximating grid.  Sampling from the discrete grid backward -- from the terminal point to the spot value of the process -- we design a simple but effective mechanism to draw Monte Carlo path and achieve a sizeable reduction of the variance associated with Monte Carlo estimators. Our conclusion is extensively supported by the numerical results presented in the final section. 

\bibliographystyle{unsrt}

\appendix

\section{The distortion function and companion parameters}\label{sec:appendix_2}
We suppose to have access to the quantizer $\Gamma_{k}$ of $\widetilde{X}_{t_{k}}$ and to the related Voronoi tessellations $\{C_{i}(\Gamma_{k})\}_{i=1,\dots,N}$. We derive an explicit expression for the distortion function $\widetilde{D}(\Gamma_{k+1})$ as follows:
\begin{equation}\label{eq:distortion_app}
  \begin{split}
    &\widetilde{D}(\Gamma_{k+1})=\mathbb{E}\left[d(\mathcal{E}_{k}(\widehat{\widetilde{X}}_{t_k}, \Delta t; Z_{t_{k+1}}),\Gamma_{k+1})^2\right]\\
    &=\sum_{i=1}^{N}\mathbb{E}\left[d(\mathcal{E}_{k}(\gamma_i^{k},\Delta t; Z_{t_{k+1}}),\Gamma_{k+1})^2\right]\mathbb{P}(\widetilde{X}_{t_k}\in C_{i}(\Gamma_{k}))\\
    &=\sum_{i=1}^{N}\sum_{j=1}^{N}(m_k(\gamma_{i}^k)-\gamma_{j}^{k+1})^2(\Phiz(\gamma_{k+1,j^{+}}(\gamma_{i}^{k}))-\Phiz(\gamma_{k+1,j^{-}}(\gamma_{i}^{k})))\mathbb{P}(\widetilde{X}_{t_k}\in C_{i}(\Gamma_{k}))\\
    &-2\sum_{i=1}^{N}\sum_{j=1}^{N}(m_k(\gamma_{i}^k)-\gamma_{j}^{k+1})v_k(\gamma_{i}^k)(\phiz(\gamma_{k+1,j^{+}}(\gamma_{i}^{k}))-\phiz(\gamma_{k+1,j^{-}}(\gamma_{i}^{k})))\mathbb{P}(\widetilde{X}_{t_k}\in C_{i}(\Gamma_{k}))\\
    &+\sum_{i=1}^{N}\sum_{j=1}^{N}v_k(\gamma_i^k)^2(\gamma_{k+1,j^{-}}(\gamma_{i}^{k})\phiz(\gamma_{k+1,j^{-}}(\gamma_{i}^{k}))-\gamma_{k+1,j^{+}}(\gamma_{i}^{k})\phiz(\gamma_{k+1,j^{+}}(\gamma_{i}^{k})))\mathbb{P}(\widetilde{X}_{t_k}\in C_{i}(\Gamma_{k}))\\
    &+\sum_{i=1}^{N}\sum_{j=1}^{N}v_k(\gamma_{i}^{k})^2(\Phiz(\gamma_{k+1,j^{+}}(\gamma_{i}^{k}))-\Phiz(\gamma_{k+1,j^{-}}(\gamma_{i}^{k})))\mathbb{P}(\widetilde{X}_{t_k}\in C_{i}(\Gamma_{k}))\,,
  \end{split}
\end{equation}
where $\Phiz$ and $\phiz$ indicate the cumulative distribution function and the probability density function of a standard Normal random variable, respectively.
To simplify notation, in Equation \eqref{eq:distortion_app}, we set for all $k\in \{0,\dots,n-1\}$ and for all $j\in \{1,\dots,N\}$
\begin{equation*}
  \begin{split}
    &\gamma_{k+1,j^{+}}(\gamma)\doteq \frac{\gamma_{j+1/2}^{k+1}-m_{k}(\gamma)}{v_{k}(\gamma)}\quad\text{and}\quad\gamma_{k+1,j^{-}}(\gamma)\doteq \frac{\gamma_{j-1/2}^{k+1}-m_{k}(\gamma)}{v_{k}(\gamma)}\,\quad\text{where}\\
    &\gamma_{j-1/2}^{k+1}\equiv \frac{\gamma_{j}^{k+1}+\gamma_{j-1}^{k+1}}{2}\,,\quad\gamma_{j+1/2}^{k+1}\equiv \frac{\gamma_{j}^{k+1}+\gamma_{j+1}^{k+1}}{2}\,,\quad\gamma_{1/2}^{k+1}\doteq -\infty\,,\quad\text{and}\quad\gamma_{N+1/2}^{k+1}\doteq +\infty\,.
  \end{split}
\end{equation*}
The so-called companion parameters $\{\mathbb{P}(\widetilde{X}_{t_{k}}\in C_{i}(\Gamma_{k}))\}_{i=1,\dots,N}$ and $\{\mathbb{P}(\widetilde{X}_{t_{k}}\in C_{j}(\Gamma_{k})\vert \widetilde{X}_{t_{k-1}}\in C_{i}(\Gamma_{k}) )\}_{j=1,\dots,N}$  are computed in a recursive way as follows:
\begin{equation*}
  \begin{split}
    &\mathbb{P}(\widetilde{X}_{t_{k}}\in C_{i}(\Gamma_{k}))=\sum_{j=1}^{N}(\Phiz(\gamma_{k,i^{+}}(\gamma_{j}^{k-1}))-\Phiz(\gamma_{k,i^{-}}(\gamma_{j}^{k-1})))\mathbb{P}(\widetilde{X}_{t_{k-1}}\in C_{j}(\Gamma_{k-1}))\,,\\
    &\mathbb{P}(\widetilde{X}_{t_{k}}\in C_{i}(\Gamma_{k})\vert \widetilde{X}_{t_{k-1}}\in C_{j}(\Gamma_{k-1}))=\Phiz(\gamma_{k,i^{+}}(\gamma_{j}^{k-1}))-\Phiz(\gamma_{k,i^{-}}(\gamma_{j}^{k-1})).
  \end{split}
\end{equation*}

\section{Lloyd I method within the \RMQA}\label{sec:appendix_1}
We present a brief review of the Lloyd I method within the Recursive Marginal Quantization framework. Let us fix $t_{k}\in \{t_{1},\dots,t_{n}\}$ and suppose we have access to the quantizer $\Gamma_{k}$ of $\widetilde{X}_{t_{k}}$ and to the associated Voronoi tessellations $\{C_{i}(\Gamma_{k})\}_{i=1,\dots,N}$. We want to quantize $\widetilde{X}_{t_{k+1}}=\mathcal{E}_{k}(\widehat{\widetilde{X}}_{t_{k}},\Delta t;Z_{t_{k+1}})$ by means of a quantizer $\Gamma_{k+1}\equiv \{\gamma_{1}^{k+1},\dots,\gamma_{N}^{k+1}\}$ of cardinality $N$. One starts with an initial guess $\Gamma_{k+1}^{0}$ and then one sets recursively a sequence $(\Gamma_{k+1}^{l})_{l\in \mathbb{N}}$ such that
\begin{equation}\label{eq:lloyd_1}
  \gamma_{j}^{k+1,l+1}=\mathbb{E}\left[\widetilde{X}_{t_{k+1}}\vert \widetilde{X}_{t_{k+1}}\in C_{j}(\Gamma_{k+1}^{l})\right]\,,
\end{equation}
where $l$ indicates the running iteration number.
One can easily check that previous equation implies that
\begin{equation*}
  q_{N}^{l+1}(\widetilde{X}_{t_{k+1}})=\mathbb{E}\left[\widetilde{X}_{t_{k+1}}\vert q_{N}^{l}(\widetilde{X}_{t_{k+1}})\right]\doteq \left(
  \mathbb{E}\left[\widetilde{X}_{t_{k+1}}\vert \widetilde{X}_{t_{k+1}}\in C_{i}(\Gamma_{k+1}^{l})\right]\right)_{1 \le i \le N},
\end{equation*}
where $q_{N}^{l}$ is the quantization associated with $\Gamma_{k+1}^{l}$. It has been proven (see \cite{pages2004optimal,carmona2012})) that $\{\|\widetilde{X}_{t_{k+1}}-q_{N}^{l}(\widetilde{X}_{t_{k+1}})\|_{2}, l\in \mathbb{N}^{+}\}$ is a non-increasing sequence and that $q_{N}^{l}(\widetilde{X}_{t_{k+1}})$ converges towards some random variable taking $N$ values as $l$ tends to infinity. From Equation~\eqref{eq:lloyd_1} and exploiting the idea of \RMQA we have
\begin{equation*}
  \begin{split}
    \gamma_{j}^{k+1,l+1}&=\mathbb{E}\left[\widetilde{X}_{t_{k+1}}\vert \widetilde{X}_{t_{k+1}}\in C_{j}(\Gamma_{k+1}^{l})\right]\\
    &=\frac{\mathbb{E}\left[\widetilde{X}_{t_{k+1}}\indicator_{\{\widetilde{X}\in C_{j}(\Gamma_{k+1}^{l})\}}\right]}{\mathbb{P}(\widetilde{X}_{t_{k+1}}\in C_{j}(\Gamma_{k+1}^{l}))}\\
    &=\frac{\mathbb{E}\left[\mathbb{E}\left[\widetilde{X}_{t_{k+1}}\indicator_{\{\widetilde{X}_{t_{k+1}}\in C_{j}(\Gamma_{k+1}^{l})\}}\vert \widetilde{X}_{t_{k}}\right]\right]}{\mathbb{E}\left[\mathbb{E}\left[\indicator_{\{\widetilde{X}_{t_{k+1}}\in C_{j}(\Gamma_{k+1}^{l})\}}\vert \widetilde{X}_{t_{k}}\right]\right]}\\
    &=\frac{\sum_{\substack{i=1}}^{N}\mathbb{E}\left[\mathcal{E}_{k}(\gamma_{i}^{k},\Delta t; Z_{t_{k+1}})\indicator_{\{\mathcal{E}_{k}(\gamma_{i}^{k},\Delta t; Z_{t_{k+1}})\in C_{j}(\Gamma_{k+1}^{l}) \}}\right]\mathbb{P}(\widetilde{X}_{t_{k}}\in C_{i}(\Gamma_{k}))}{\sum_{i=1}^{N}\mathbb{P}(\mathcal{E}_{k}(\gamma_{i}^{k},\Delta t; Z_{t_{k+1}})\in C_{j}(\Gamma_{k+1}^{l})\mathbb{P}(\widetilde{X}_{t_{k}}\in C_{i}(\Gamma_{k}))}\,.
  \end{split}
\end{equation*}
The last term in previous equation is equivalent to the stationary condition in Equation \eqref{eq:stationarity} for the quantization $q_{N}(\widetilde{X}_{t_{k+1}})$. Then, the stationary condition is equivalent to a fixed point relation for the quantizer.

\section{Robustness checks}\label{sec:appendix_3}
We test the convergence of Lloyd I method with and without Anderson acceleration on the quantization of a standard Normal random variable\footnote{At \url{www.quantize.maths-fi.com} a database providing quadratic optimal quantizers of the standard univariate Gaussian distribution from level $N=1$ to $N=1000$ is available.} initialised from a distorted quantizer. We indicate by $\Gamma_{\mathcal{N}(0,1)}^{*}$ the optimal quantizer of a standard Normal random variable and we distort it through the multiplication by a constant $c$, that is to say $c\times \Gamma_{\mathcal{N}(0,1)}^{*}$. Then, we monitor the convergence of both algorithms to $\Gamma_{\mathcal{N}(0,1)}^{*}$ starting from $c\times \Gamma_{\mathcal{N}(0,1)}^{*}$. 
The error at iteration $l$ is defined as $\|\Gamma_{\mathcal{N}(0,1)}^{*}-\Gamma_{\mathcal{N}(0,1)}^{l}\|_{2}$, with $\Gamma_{\mathcal{N}(0,1)}^{l}$ the quantizer found by the algorithms at the $l$-th iteration and $\|\cdot\|_{2}$ the Euclidean norm in $\mathbb{R}^{N}$. The stopping criteria is set to $\|\Gamma_{\mathcal{N}(0,1)}^{l+1}-\Gamma_{\mathcal{N}(0,1)}^{l}\|_{2}\le 10^{-7}$, the level of the quantizer to $N=10$, and the constant $c$ to 1.01. The results of our investigation are summarized in Figure~\ref{fig:acceleration_noacceleration}. We can graphically assess the rate of convergence\footnote{We recall that a sequence $(\Gamma^{l})_{l\in \mathbb{N}}$ converging to a $\Gamma^{*}\neq \Gamma^{l}$ for all $l$ is said to converge to $\Gamma^{*}$ with order $\alpha$ and asymptotic error constant $\lambda$ if there exist positive constants $\alpha$ and $\lambda$ such that
\begin{equation*}
  \lim_{l\rightarrow \infty}\frac{\|\Gamma^{l+1}-\Gamma^{*}\|_{2}}{\|\Gamma^{l}-\Gamma^{*}\|_{2}^{\alpha}}=\lambda\,.
\end{equation*}
} 
for both algorithms.  In case of Lloyd I method without acceleration the convergence is, as expected, linear. For Lloyd I method with acceleration the rate is not well defined, but Figure~\ref{fig:acceleration_noacceleration} shows the impressive improvement in the convergence towards the known optimal quantizer.
\begin{figure}[t]
  \centering
  \includegraphics[scale=1]{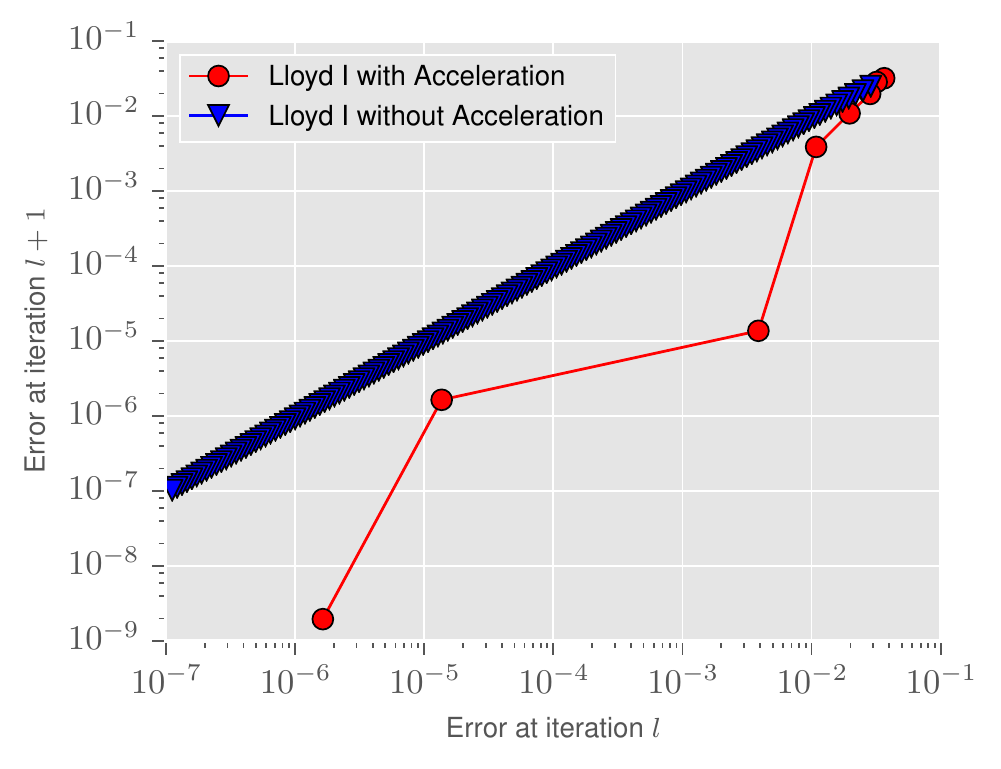}
  \caption{Quantization of a standard Normal random variable: Comparison of the convergence of Lloyd I method with and without Anderson acceleration.}
  \label{fig:acceleration_noacceleration}
\end{figure}
Figure~\ref{fig:iterations} supports the same conclusion in terms of the number of iterations necessary to reach the stopping criterion.
\begin{figure}[t]
  \centering
    \includegraphics[scale=1]{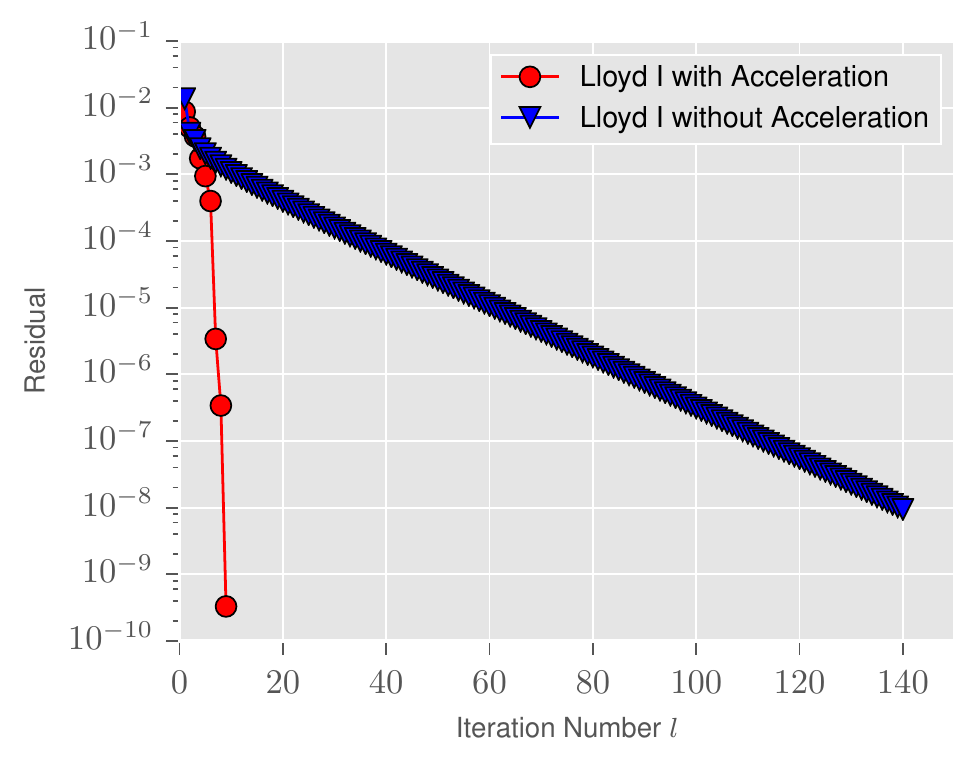} 
  \caption{Quantization of a standard Normal random variable: Comparison of the number of iterations necessary to reach the stopping criterion between Lloyd I method with and without Anderson acceleration.}
  \label{fig:iterations}
\end{figure}

Then, we investigate numerically the sensitivity of Lloyd I method with Anderson acceleration and Newton-Raphson algorithm to the initial guess as a function of the distortion $c$ applied to the optimal quantizer $\Gamma_{\mathcal{N}(0,1)}^{*}$. The results of our investigation are summarized in Figure \ref{fig:dependence_distortion}. The four panels correspond to different levels of distortion $c=\{1.10, 1.20, 1.25, 1.35\}$. As before, we set $N=10$ whereas on the $y$ axis we report the residual at iteration $l$, $\|\Gamma^{l+1}-\Gamma^{l}\|_{2}$. For low levels of distortion Newton-Raphson method converges to the optimal solution more quickly than Lloyd I method. This result confirms the theoretical behavior due to the quadratic rate of convergence of the Newton-Raphson algorithm. However, when the initial guess is quite far from the solution -- as it is for the cases of 25\% and 35\% distortion -- the algorithm may spend many cycles far away from the optimal grid.    
\begin{figure}[ht]
  \centering
  \begin{minipage}{0.495\linewidth}
    \includegraphics[scale=0.75]{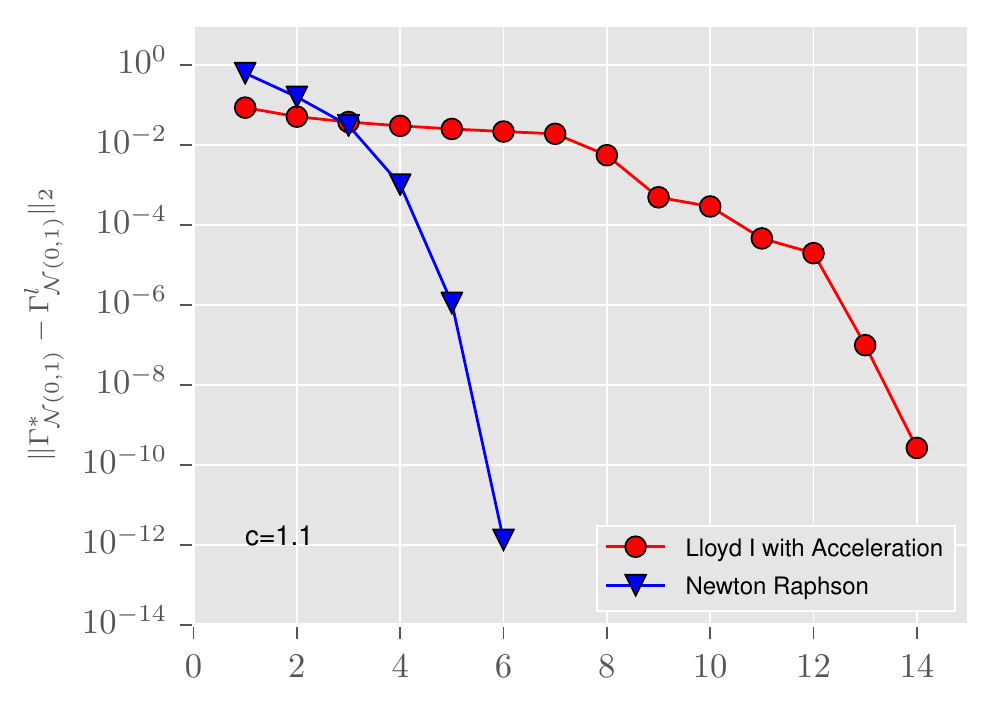} 
  \end{minipage}
  \begin{minipage}{0.495\linewidth}
    \includegraphics[scale=0.75]{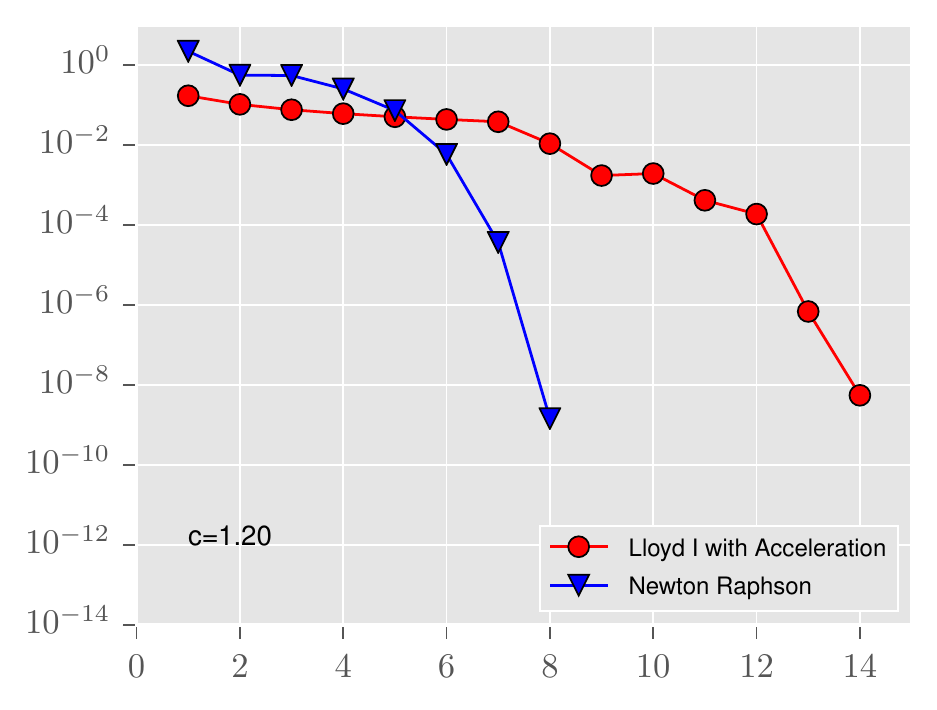} 
  \end{minipage}
  \begin{minipage}{0.495\linewidth}
    \includegraphics[scale=0.75]{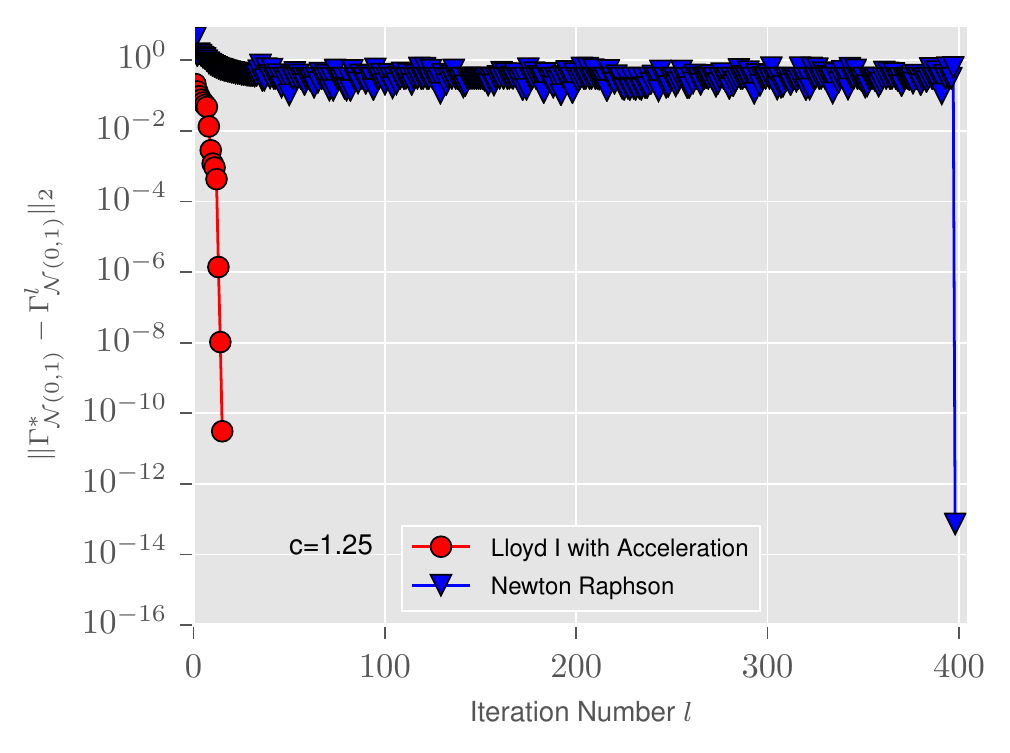} 
  \end{minipage}
  \begin{minipage}{0.495\linewidth}
    \includegraphics[scale=0.75]{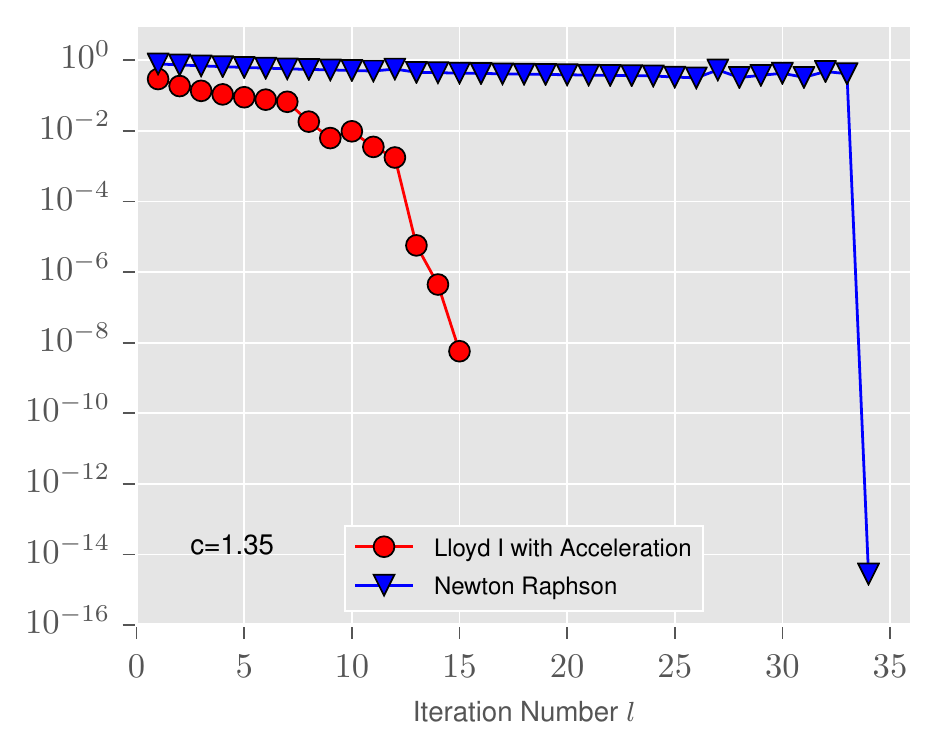} 
  \end{minipage}
  \caption{Quantization of a standard Normal random variable: Comparison between Lloyd I method with Anderson acceleration and Newton-Raphson algorithm.}
  \label{fig:dependence_distortion}
\end{figure}

Finally, we examine the convergence of Lloyd I and Newton-Raphson algorithms when considering the following Euler-Maruyama discrete scheme 
\begin{equation*}
  \begin{cases}
    \bar{X}_{t_{k+1}}=\bar{X}_{t_{k}}+r\bar{X}_{t_{k}}\Delta t+\sigma \bar{X}_{t_{k}}\sqrt{\Delta t}~Z_{t_k}\,,\\
    X_{0}=x_{0}\,,
  \end{cases}
\end{equation*}
with $r$, $\sigma$, and $x_0$ strictly positive real constants, and $\Delta t = t_{k+1} - t_k$ for all $k=1,\ldots,n-1$. To enlighten the greater sensitivity of the Newton-Raphson method to the grid initialisation in comparison with the Lloyd I with Anderson acceleration it is sufficient to stop at $n=2$ with $\Delta t = 0.01$. We set the level of the quantizers $\Gamma_1$ and $\Gamma_2$ equal to $N=30$ and $x_0=1$. By definition the random variable $\bar{X}_{t_{1}}\sim\mathcal{N}(m_{0}(x_{0}),v_{0}(x_{0}))$ where $m_{0}(x_{0}) =  x_{0}+rx_{0}\Delta t$ and $v_{0}(x_{0}) = \sigma x_{0}$. In order to compute the quantizer for $\bar{X}_{t_{1}}$ we initialise the algorithms at time $t_{1}$ to $m_{0}(x_{0})+v_{0}(x_{0})\Gamma_{\mathcal{N}(0,1)}^{*}$, with $\Gamma_{\mathcal{N}(0,1)}^{*}$ the optimal quantizer of a standard Normal random variable. Once we have obtain the optimal quantizer $\Gamma_{1}^{*}=\{\gamma_1^{*1},\cdots,\gamma_1^{*30}\}$ we set the initialisation of the quantizer $\Gamma_{2}^{Init}=\{\gamma_2^1,\cdots,\gamma_2^{30}\}$ at time $t_{2}$ using one of the following alternatives
\begin{itemize}
  \item[i.] the optimal quantizer at the previous step
    \begin{equation*}
      \Gamma_{2}^{Init}= \Gamma_{1}^{*}\,;
    \end{equation*}
  \item[ii.] the Euler operator
    \begin{equation*}
      \gamma_{2}^{i} = m_{1}(\gamma_{1}^{i})+v_{1}(\gamma_{1}^{i}) \Gamma_{\mathcal{N}(0,1)}^{*,i}\,,
    \end{equation*}
    for $i = 1,\dots,30$;
  \item[iii.] the mid point between Euler operator and the optimal quantizer at the previous step
    \begin{equation*}
      \gamma_{2}^{i} = 0.5 \gamma_{1}^{i}+0.5 (m_{1}(\gamma_{1}^{i})+v_{1}(\gamma_{1}^{i}) \Gamma_{\mathcal{N}(0,1)}^{*,i})\,,
    \end{equation*}
    for $i = 1,\dots,30$;
  \item[iv.] the expected value
    \begin{equation*}
      \gamma_{2}^{i} = m_{1}(\gamma_{1}^{i})\,,
    \end{equation*}
    for $i = 1,\dots,30$.
\end{itemize}
The left panel of Figure~\ref{fig:initialisation} shows the four different initial grids which correspond to above specifications. In the same panel, on the right side, we also plot the optimal quantizer $\Gamma_2^*$ to which both Lloyd I with Anderson acceleration and Newton-Raphson methods should converge. The right panel report the quantization error -- defined as $\sqrt{\widetilde{D}_{2}(\Gamma_{2}^{l})}$ -- as a function of the iteration number $l$. We stop the algorithm when the residual falls below $10^{-5}$. The numerical investigation shows that the Newton-Raphson method converges to the optimal grid faster than the Lloyd I method,  with the only exception represented by the case $\Gamma_2^{init}$ equal to the mid point. However, when initialized with the Euler operator or the mid point Newton-Raphson algorithm fails to converge due to the bad condition number of the Hessian matrix (corresponding lines are not reported on the Figure). This result is in line with the findings in~\cite{pages2003numerics} where the authors stress that the Newton-Raphson method may fail even for symmetric initial vectors since the anomalous behavior of some components of the Voronoi tessellation.

\medskip

In light of above explorations, we finally conclude that the approach based on a fixed-point algorithm such as the Lloyd I method with Anderson acceleration is much more robust than a Newton-Raphson approach -- which in the present application relies on the computation of the Hessian of the matrix.

\begin{figure}[t]
  \centering
  \begin{minipage}{0.495\linewidth}
    \includegraphics[scale=0.75]{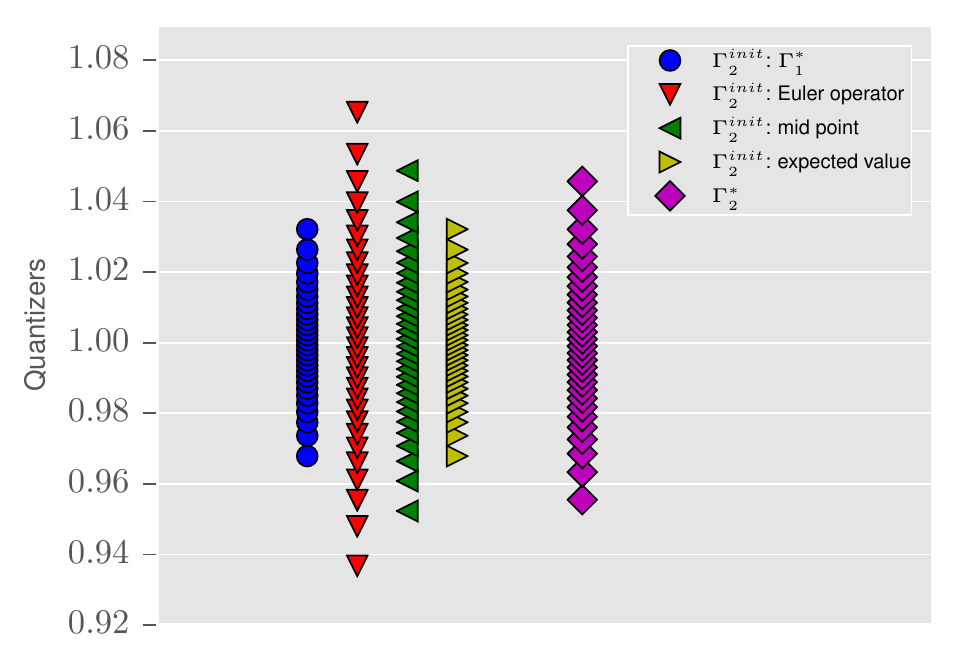} 
  \end{minipage}
  \begin{minipage}{0.495\linewidth}
    \includegraphics[scale=0.75]{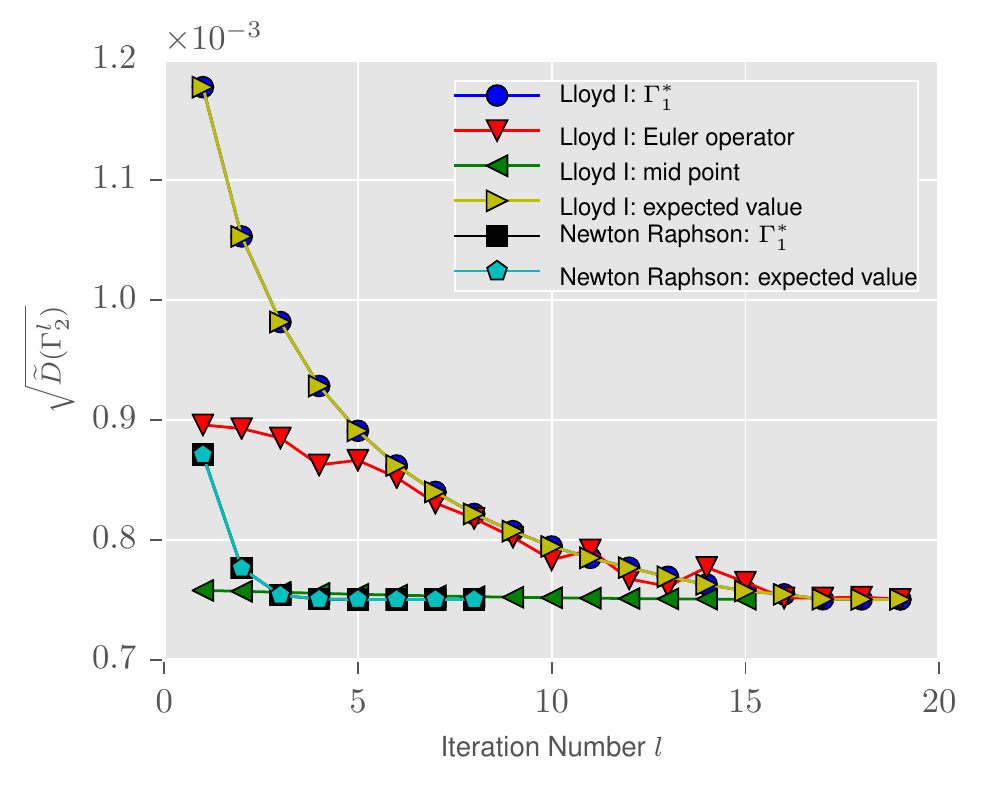} 
  \end{minipage}
  \caption{Quantization of the Euler-Maruyama scheme associated to a Geometric Brownian motion. Left panel: four different initial grids $\Gamma_2^{init}$ and optimal grid $\Gamma_2^*$ on the left and right sides, respectively. Right panel: quantization error as a function of the iteration number.}
  \label{fig:initialisation}
\end{figure}

\section{Construction of the Markov generator $\mathcal{L}_{\Gamma}$}\label{sec:appendix_4}
We define $\Delta \gamma \doteq \gamma_{i+1}-\gamma_{i}$, $1\le i \le N-1$. The finite difference approximation of the first and second partial derivative in Equation \eqref{eq:markov_generator} is defined as
\begin{equation*}
  \begin{split}
    &\frac{\partial u}{\partial \gamma}(\gamma,t)\approx \frac{u(\gamma_{i}+\Delta \gamma,t)-u(\gamma_{i}-\Delta \gamma,t)}{2\Delta \gamma}\,,\\
    &\frac{\partial^{2} u}{\partial \gamma^{2}}(\gamma,t)\approx \frac{u(\gamma_{i}+\Delta \gamma,t)-2 u(\gamma_{i},t)+ u(\gamma_{i}-\Delta \gamma, t)}{(\Delta \gamma)^2}\,,
  \end{split}
\end{equation*}
for all $t \in [0,T]$.
The Markov generator $\mathcal{L}_{\Gamma}$ is the $N\times N$ matrix defined as
\begin{equation*}
\mathcal{L}_{\Gamma} \doteq
 \begin{pmatrix}
  d_1 & u_1  & 0   & 0 &\cdots & 0 &0\\
  l_2 & d_2  & u_2 & 0 &\cdots & 0 &0\\
   0  &\ddots& \ddots& \ddots  & \cdots&  0&0\\
   0  & 0    &  l_i  & d_i & u_i&0&  0\\
   0  & 0    & 0&\ddots& \ddots& \ddots  &  0 \\
   0  & 0    & 0& 0    &l_{N-1}& d_{N-1} &  u_{N-1} \\
   0  & 0    & 0& 0 & 0& l_{N}  &  d_{N}
 \end{pmatrix}\,,
\end{equation*}
where the coefficients $l_{i}$, $d_{i}$ and $u_{i}$ are given by
\begin{equation*}\label{eq:coefficients}
  \begin{split}
    &l_{i}=-\frac{b(\gamma_{i})}{2\Delta \gamma}+\frac{1}{2}\frac{\sigma(\gamma_{i})^2}{(\Delta \gamma)^2}\,,\\
    &d_{i}=-\frac{\sigma(\gamma_{i})^2}{(\Delta \gamma)^2}\,,\\
    &u_{i}=+\frac{b(\gamma_{i})}{2\Delta \gamma}+\frac{1}{2}\frac{\sigma(\gamma_{i})^2}{(\Delta \gamma)^2}\,,
  \end{split}
\end{equation*}
for all $1\le i \le N$. The coefficients of the first and last row are chosen so that the Markov chain is reflected at the boundaries of the state domain. The choice of the boundary conditions should have a negligible effect provided that the range of the state domain is sufficiently large.
\end{document}